\documentclass[12pt]{article}
\usepackage[margin=1in]{geometry}
\usepackage{times}
\usepackage{soul}
\usepackage{url}
\usepackage[hidelinks]{hyperref}
\usepackage[utf8]{inputenc}
\usepackage[T1]{fontenc}
\usepackage{mathptmx}
\usepackage[small]{caption}
\usepackage{graphicx}
\usepackage{amsmath}
\usepackage{amssymb}
\usepackage{amsthm}
\usepackage{booktabs}
\usepackage{multirow}
\usepackage{algorithm}
\usepackage{microtype}
\usepackage{algpseudocode}
\usepackage{enumitem}
\usepackage{thm-restate}
\usepackage{thmtools}
\usepackage{tcolorbox}
\usepackage{cleveref}

\newcommand{\Z}{\mathbb{Z}}

\newcommand{\cT}{\mathcal{T}}

\newtheorem{observation}{Observation}
\newtheorem{corollary}{Corollary}
\newtheorem{lemma}{Lemma}
\newtheorem{definition}{Definition}
\newtheorem{proposition}{Proposition}
\newtheorem{claim}{Claim}
\newtheorem{rr}{Reduction Rule}

\usepackage{tikz}
\usetikzlibrary{fit, arrows.meta, positioning, backgrounds}
\usetikzlibrary{decorations.pathreplacing, calligraphy, arrows.meta}

\newcommand{\SR}[1]{\textcolor{red}{}}

\date{}

\title{Complexity of Eliminating (Majority) Illusion in Directed Networks}

\author{
Sougata Jana,
Sanjukta Roy\\
Indian Statistical Institute Kolkata\\
sougatajana\_r@isical.ac.in, sanjukta@isical.ac.in
}

\begin{document}
\maketitle
\begin{abstract}
We study illusion elimination problems on directed social networks where each vertex is colored either red or blue. A vertex is under \textit{majority illusion} if it has more red out-neighbors than blue out-neighbors when there are more blue vertices than red ones in the network. In a more general phenomenon of $p$-illusion, a vertex under illusion has at least $p$ fraction of the out-neighbors (as opposed to $1/2$ for majority) of a vertex is red. In the directed illusion elimination problem, we recolor minimum number of vertices so that no vertex is under $p$-illusion, for $p\in (0,1)$. Unfortunately, the problem is NP-hard for $p =1/2$ even when the network is  a grid. Moreover, the problem is NP-hard and W[2]-hard when parameterized by the number of recoloring for each $p \in (0,1)$ even on bipartite DAGs. Thus, we can neither get a polynomial time algorithm on DAGs, unless P=NP, nor we can get a FPT algorithm even by combining solution size and directed graph parameters that measure distance from acyclicity, unless FPT=W[2]. Furthermore, the problem is NP-complete in graphs of degree at most three for $p\in (0,\frac{2}{3}]$. We show that the problem can be solved in polynomial time in structured, sparse networks such as outerplanar networks, outward grids, trees, and cycles. Finally, we show tractable algorithms parameterized by treewidth of the underlying undirected graph and maximum deficiency of the graph combined, and by the number of vertices under illusion.
\end{abstract}

\section{Introduction}

    % Many fundamental problems in combinatorial optimization can be expressed as selecting a small subset of vertices to satisfy local constraints. Classic examples include vertex cover, where each edge must be covered by at least one selected endpoint, and hitting set, where each set must contain at least one selected element. Such problems play a central role in complexity theory, approximation algorithms, and parameterized complexity. 
    
    Social networks influence opinions in both practice and theory \cite{castiglioni2020election,ferrara2022twitter,FujiwaraMS24}. Our perceptions are heavily filtered by those we follow rather than the objective global truth.
    % It can enforce existing bias, create pluralistic ignorance, information bubble, and political polarization.
    This vulnerability is often exploited by political parties or marketing brands through \textit{majority illusion} \cite{lerman2016majority}, a social phenomenon where a rare opinion appears mainstream because a few highly connected agents adopt it. 
    
    The implications of majority illusion extend significantly into strategic communication, particularly within the realms of political campaigning, digital advertising, and the optimization of influence \cite{zhou2021maximizing,baker1994approximation}. Furthermore, the perceptual distortions due to illusion impact voting behavior \cite{wilder2018controlling,doucette2019inferring} and various collective decision-making processes \cite{bara2021predicting,castiglioni2020election}, as well as the equitable distribution of resources through participatory budgeting \cite{kempe2020inducing}. Johnson et al.~\cite{Johnson2020} showed the effects of majority illusion in social network in the decision for getting vaccinated. %, highlighting that AI methods for network analysis and information diffusion are crucial in understanding and countering the spread of misinformation in social networks. %The study shows that smaller but well-connected adversarial communities (anti-vaccine) can dominate larger but peripheral communities (pro-vaccine). 
Schrama et al.~\cite{schrama2025majority} showed that majority illusion can push the entire population from an abundant, high-welfare state to a depleted, low-welfare one. Both works highlight the need for network analysis and network-based policies to prevent misconception driven convergence to a undesirable state.

We study the complexity of eliminating illusion in social networks by modifying minimum number of agents' opinions.  The computational complexity of illusion elimination depends heavily on network topology.
Fioravantes et al.~\cite{fioravantes2025eliminating} studied networks as undirected graphs and considered  the problem of eliminating majority illusion. 
However, in social networks, usually, the influence between the agents are not symmetric, i.e., the edge relations are not symmetric. Thus, we model a network as a directed graph $G$ where the vertex set $V$ are the set of agents and an \textit{edge} from agent $i$ to agent $j$ implies that $i$ can be \textit{influenced by} $j$. Moreover, in directed networks, it is natural to assume that agents are influenced by their out-neighbors and may change opinion due to lobbying or bribery \cite{bredereck2017manipulating,faliszewski2022opinion}.
An agent's opinion is represented by a coloring function $f : V \to \{B,R\}$, where $B$ and $R$ denotes blue and red, respectively. For example, the two colors can represent a two-party election.\footnote{It does not restrict the model to the two party case. When there are multiple parties and one party is majority, then the majority can be represented by blue and all the other parties by red.} If there are more blue voters than red, then we say that blue is the majority. A voter is under \textit{majority illusion} if it has strictly more red out-neighbors than blue out-neighbors. 
    % This formulation is motivated by situations where agents make binary decisions\footnote{The two colors can represent a two-party election. However, it does not restrict the model to the two party case. When there are multiple parties and one party is majority, then the majority can be represented by blue and all the other parties by red.}, and edges encode asymmetric influence between the agents. 
    There can be instances where some or all vertices are under majority illusion (see \Cref{fig:1}(a)). 

\begin{figure}[t]
    \centering
    % First diagram - changed [b] to [t]
    \begin{minipage}[t]{0.49\linewidth}
        \centering
        \begin{tikzpicture}[
    % Styles for the nodes
    variable_node/.style={circle, draw, minimum size=0.3cm, inner sep=0pt, thick},
    rednode/.style={circle, draw=red, fill=red!5, minimum size=0.3cm, thick},
    bluenode/.style={circle, draw=blue, fill=blue!5, minimum size=0.3cm, thick}, 
    font=\small,
    node distance=0.4cm % Set a standard horizontal gap
]
        % \begin{tikzpicture}[->, >=Stealth,scale=1, every node/.style={circle, draw, minimum size=4mm, inner sep=0.2pt, font=\tiny}]
            \node[bluenode, label=above:{$v_1$}] (v1) at (7,0.75) {} ;
            \node[bluenode, label=above:{$v_2$}] (v2) at (10,0.75) {} ;
            \node[bluenode, label=below:{$v_3$}] (v3) at (10,0) {} ;
            \node[bluenode, label=right:{$v_4$}] (v4) at (8.69,0.45) {} ;
            \node[bluenode, label=below:{$v_5$}] (v5) at (7,0) {} ;
            \node[rednode, label=above:{$v_6$}] (v6) at (8,0.75) {}  ;
            \node[rednode, label=above:{$v_7$}] (v7) at (9.3,0.75) {}  ;
            \node[rednode, label=below:{$v_8$}] (v8) at (9.3,0) {} ;
            \node[rednode, label=below:{$v_9$}] (v9) at (8,0) {} ;

            \draw[-Stealth] (v1) -- (v6);
            \draw[-Stealth] (v2) -- (v7);
            \draw[-Stealth] (v3) -- (v7);
            \draw[-Stealth] (v3) -- (v8);
            \draw[-Stealth] (v4) -- (v8);
            \draw[-Stealth] (v4) -- (v9);
            \draw[-Stealth] (v5) -- (v9);
            \draw[-Stealth] (v5) -- (v6);
            \draw[-Stealth] (v6) -- (v7);
            \draw[-Stealth] (v7) -- (v8);
            \draw[-Stealth] (v8) -- (v9);
            \draw[-Stealth] (v9) -- (v6);
        \end{tikzpicture}
        \caption*{(a) Every vertex is under majority illusion: has more red out-neighbors than blue}
    \end{minipage}
    \hfill
    % Second diagram - changed [b] to [t]
    \begin{minipage}[t]{0.49\linewidth}
        \centering
        \begin{tikzpicture}[
    % Styles for the nodes
    variable_node/.style={circle, draw, minimum size=0.3cm, inner sep=0pt, thick},
    rednode/.style={circle, draw=red, fill=red!5, minimum size=0.3cm, thick},
    bluenode/.style={circle, draw=blue, fill=blue!5, minimum size=0.3cm, thick}, 
    font=\small,
    node distance=0.4cm % Set a standard horizontal gap
]
            \node[bluenode, label=below:{$v_1$}] (v1) at (0,0) {};
            \node[bluenode, label=below:{$v_2$}] (v2) at (1,0) {};
            \node[rednode, label=above:{$v_3$}] (v3) at (0,1) {};
            \node[rednode, label=above:{$v_4$}] (v4) at (1,1) {};

            \draw[-Stealth] (v1) -- (v3);
            \draw[-Stealth] (v2) -- (v3);
            \draw[-Stealth] (v1) -- (v4);
            \draw[-Stealth] (v2) -- (v4);
        \end{tikzpicture}
        \caption*{(b) Blue vertices are under $p$-illusion for every $p\in (0, 1)$ as they have only red out-neighbors}
    \end{minipage}   
    \caption{Illusions in different networks}
    \label{fig:1}
\end{figure}

In this paper, we initiate the study of the computational question of eliminating illusion in directed graphs by recoloring minimum number of vertices. We assume that blue is the majority. A graph is \emph{majority illusion-free} if none of its vertices is under majority illusion.
 % We study the \textsc{Directed  Illusion-free Recoloring({Difr})} problem where the objective is to recolor a minimum number of red vertices to eliminate illusion.
 We ask, can $G$ be illusion-free 
% \SR{define illusion-free in the previous para} 
by recoloring at most $k$ vertices?
 In a \textit{minimum illusion-free recoloring}, a minimum number of vertices is recolored such that the graph $G$ is illusion-free.
\begin{tcolorbox}[colback=gray!10!white,    % background color (light grey)
                  colframe=black,           % border color
                  boxrule=0.4pt,            % border thickness
                  arc=1mm,                  % rounded corners
                  width=\linewidth]         % full text width
{\textsc{Directed Illusion-Free Recoloring (Difr)}:}\\
\textbf{Input:} \( G = (V, E) \), a directed graph, a coloring function \( f : V \rightarrow \{ B, R \} \), and a nonnegative integer \( k\).\\
\textbf{Question:} Does there exist a (recoloring) function \( f' : V \rightarrow \{ B, R \} \) of $f$ such that $G$ is illusion-free under the coloring $f'$, where $| \{ v \mid f(v) \neq f'(v) \} | \leq k?$
\end{tcolorbox}

% {\color{red} CHANGE}An agent's perception may change not only based on the majority (i.e., $1/2$ of its out-neighbors), but also due to a $p$ fraction of the out-neighbors for some $p \in [0,1]$, as suggested by diffusion models like linear threshold \cite{granovetter1978threshold,schelling2006micromotives}. 
Majority illusion is often introduced in the context of elections. However, applications like vaccination provides a more compelling motivation to look beyond majority.
Consider a contact network where each agent is colored according to whether they are vaccinated. For many contagious diseases, achieving herd immunity requires a high global vaccination rate (often well above \(50\%\), and in some cases exceeding \(90\%\)) \cite{fine2011herd,anderson1992infectious}. This corresponds to a super-majority i.e., we would require almost $90\%$ of a neighborhood should be vaccinated. Conversely, when the goal is accurate perception of underrepresented groups, it is relevant to consider the status of a subset of neighborhood that is strictly smaller than majority. Network structure can systematically distort individuals’ beliefs about the prevalence of minority attributes (e.g., race, sexual orientation, disability status), leading to misperceptions that are well documented in studies of social networks and the majority illusion \cite{lerman2016majority,feld1991your}. These applications require a definition of illusion more general than majority illusion. We define illusion based on a fraction $p$ that takes value $>1/2$ for the former example or a value $<1/2$ for the later example.

We define a vertex to be \textit{under $p$-illusion} if strictly less than $p$ fraction of its out neighbors are blue (see \Cref{fig:1}(b)). We study elimination of $p$-illusion by editing the vertex colors in directed graphs where 
 % Thus, we generalize our \textsc{Difr} problem and consider the problem of $p$-\textsc{Difr}.
 we do not impose the assumption that the blue vertices form majority. We ask, for a given rational number $p \in [0, 1]$, is it possible to ensure that every voter has at least $p$-fraction of blue out-neighbors?  %A vertex $v$ is $p$-\textit{illusion-free} if it does not have $p$-illusion. %the number of blue out-neighbors of $v$ is at least $p$-fraction of its total out-neighbors. We define the problem $p$-\textsc{Difr} as follows. 
A graph $G$ is \textit{$p$-illusion-free} if each vertex in the graph is $p$-illusion-free. Given a rational $p \in [0,1]$, we define the following problem.
 \begin{tcolorbox}[colback=gray!10!white,    % background color (light grey)
                  colframe=black,           % border color
                  boxrule=0.4pt,            % border thickness
                  arc=1mm,                 % rounded corners
                  width=\linewidth]         % full text width
{\textsc{Directed $p$-Illusion-Free Recoloring} ($p$-\textsc{Difr}):}\\
\textbf{Input:} \( G = (V, E) \), a directed graph, a coloring function \( f : V \rightarrow \{ B, R \} \), and a nonnegative integer \( k\).\\ %, and a rational number $p \in [0, 1]$.\\
\textbf{Question:} Does there exist a (recoloring) function \( f' : V \rightarrow \{ B, R \} \) of $f$ such that $G$ is $p$-illusion-free under the coloring $f'$, where $| \{ v \mid f(v) \neq f'(v) \} | \leq k?$
\end{tcolorbox}

Note that $p=\frac{1}{2}$ with blue as the strict majority is the \textsc{Difr} problem. The $p$-illusion elimination problem has been studied for undirected graphs where a certain number of edges can be added and/or deleted~\cite{dippel2025eliminating}. Eliminating illusions by editing edges in directed graphs is trivial in the sense that it does not utilize the topology of the network and depends only on deficiency of the vertices. 

\subsection{Our Contribution.}
We initiate the study of eliminating $p$-illusion by recoloring vertices. Furthermore, we initiate the analysis of computational complexity of eliminating majority illusion and $p$-illusion in directed graphs.
 % (i) we generalize the question eliminating majority illusion to eliminating $p$-illusion and analyze their computational complexity in directed graphs. 
% We begin by formulating the problem of eliminating majority illusion and, in general, eliminating $p$-illusion by recoloring vertices in directed graphs. This is a natural model of a social network.
% It is not clear apriori whether it would require more recolorings to eliminate $p$-illusion, as opposed to the majority illusion in undirected version, due to the non-symmetric edges.
In directed graphs, when we recolor a vertex, only its in-neighbors can see it, the out-neighbors remain unaffected, unlike the undirected case. Thus, it may feel that a minimum illusion-free recoloring can be found easily. Indeed, that is the case when we add and/or delete edges of the graph to eliminate majority illusion.\footnote{For example, when we are adding edges to eliminate illusion, it does not depend on the graph, the number of edges we need to add is total `deficiency' of the vertices under illusion, where deficiency of a vertex is the number of extra red out-neighbors.}  However, we show that the recoloring problem is harder.
We show that the complexity of the problem depend on three aspects: graph class, graph parameters, and the value of $p$.  We systematically and extensively study the
algorithmic complexity through a combination of them. 

\textbf{Graph Classes.} We primarily focus on cycles, grids, planar graphs, trees, and tree-like graphs. Note that as the influence between the agents can come from their residential neighborhood, or can be constrained by spatial or local interactions, it is natural to assume that the network of social interactions to be a planar graph or simply a grid. Similarly, a directed tree or out-tree (also known as arborescence) models hierarchical influence, such as mentorship, organizational chains, or follower cascades, where influence flows outwards from central sources. 

\textbf{Parameters.} We specifically look at two structural parameters: maximum degree and tree-width of the graph; and the problem specific parameters: solution size $k$, number of vertices under illusion $|X|$, and maximum deficiency of any vertex $D$. Deficiency indicates how far a vertex is from being illusion-free. Parameter maximum degree measures the maximum number of agents that are known to an agent and is motivated by the observation that each agent typically only knows a few other agents, and thus, if people only follow who they know, then it yields a bounded degree graph~\cite{bachrach2013optimal,CsehIrvingManlove2019}. As we show that the problem is polynomial time solvable for trees, the parameter tree-width measures a distance from easy instances. Also, it is a standard parameter to understand the structure of the graph. Furthermore, as the brute-force method would give an XP algorithm parameterized by the solution size, it is natural to ask if that is the best we can do. Deficiency $D$ measures the gap in perception that caused the illusion. Thus, we ask if the problem is easy when there are small number of vertices under illusion or their illusion is caused by small gap in perception.

\textbf{Values of $p$.} As we focus on the computational question of eliminating $p$-illusion, unlike previous works on this problem, we need to understand the role of $p$ in the complexity of it. It can be easily observed that for large values of $p$, e.g., in a graph with maximum degree $\Delta$, when $p>\frac{\Delta-1}{\Delta}$, all the out-neighbors of each vertex must be blue, consequently, the problem is trivial. However, small values of $p$, does not create such easy instances.

% \begin{itemize}[wide, listparindent=\parindent, noitemsep,topsep=0pt,label=--]%[wide, listparindent=\parindent, parsep=\parskip,label=--]

We begin by showing that  \textsc{Difr} is NP-complete on directed grids. %That is, it remains NP-hard on simple graphs like grids where each vertex has at most four neighbors.
 Then, we show that solving $p$-\textsc{Difr} is NP-hard even for graphs of degree at most $3$ for $p\in (0,\frac{2}{3}]$. Towards this we first show that it is NP-hard on graphs of maximum degree $3$ when $p\in(1/3,2/3]$. %Then we prove the following more general result.
% that the problem of eliminating $p$-illusion is NP-complete and W[2]-hard with respect to $k$, the number of recolorings, even on  directed acyclic bipartite graphs (DAGs) for each rational value of $p \in (0,1)$. 

% To establish \Cref{thm:pdifrreduction}, we provide a parameter preserving reduction from \textsc{Hitting Set} to our problem. \textsc{Hitting Set} is  a well-known NP-complete~\cite{garey1990computers} problem and it is  W[2]-hard paramatereized by solution size~\cite{cygan2015parameterized}. We impose a lower bound of $\frac{1}{p}$ on the sizes of the sets in the input family. The problem remains W[2]-hard under this restriction. 
% In the construction, the out-neighborhoods of the vertices are structurally designed so that the resulting $p$-deficiency never exceeds one, thereby ensuring the reduction works as intended.
% When $3$-\textsc{Hitting Set} is used as the source problem, a slight modification of the reduction shows that $p$-\textsc{Difr} is NP-hard for graphs of maximum degree $3$ for $p \in (0,1/3]$ (\Cref{prop:deg3reductionfor p zero to one-third}).

Furthermore, we prove a more general result. We show that
    $p$-\textsc{Difr} is NP-complete and W[2]-hard for any $p\in (0,1)$ even on acyclic bipartite graphs. 
This implies that we cannot hope to obtain a fixed parameter tractable (FPT) algorithm parameterized by $k$, i.e., it cannot be solved in time $f(k)|I|^{\mathcal{O}(1)}$, where $f$ is any computable function of $k$, and $|I|$ denotes the input size, unless W[2] = FPT. 
  
 In the construction of all of our hardness reductions, each vertex under illusion needs exactly one additional blue out-neighbor to be illusion-free, i.e., we say that the maximum deficiency of any vertex is at most one. Thus, $p$-\textsc{Difr} is paraNP-hard \footnote{The problem is NP-hard even when the parameter is constant, i.e., we cannot expect to get even an XP algorithm with respect to this parameter, unless P= NP} with respect to maximum deficiency of any vertex. Moreover, since it is hard on DAGs, it is paraNP-hard with respect to distance from acyclicity, including many directed graph parameters such as directed feedback vertex set, feedback arc set, directed treewidth, directed cutwidth, digraph  bandwidth, or degreewidth, to name a few.

To understand the true causes of the intractability results, we explore the line between easy and hard graph topologies.
% \begin{itemize}[wide, listparindent=\parindent, noitemsep,topsep=0pt,label=--]
 We show a dichotomy result as we prove that the $p$-\textsc{Difr} is polynomial time solvable when the maximum degree $ \Delta \leq 2$.
  We identifying orientations of grid graphs where eliminating $p$-illusion is solvable in polynomial time.
We call them as \emph{outward grid graphs}, where the edges are oriented from left to right and from top to bottom. 
 Outward grids naturally model hierarchical or layered information flow, where influence propagates directionally from central sources to peripheral nodes. We establish that $p$-\textsc{Difr} on $G$  can be solved in time $\mathcal{O}(mn\sqrt{mn})$ on an $m\times n$ outward grid.

 % To prove \Cref{thm:outgrid}, we first recolor the out-neighbors of the vertices under illusion that have out-degree one. Then, we construct  an auxiliary graph $H$ which is undirected, acyclic, and bipartite. Finally, we show that an optimal solution can be obtained by recoloring a minimum size vertex cover of $H$.
 
 % $H=(R(V),E')$, where $R(V)$ is the set of red vertices, and
 %    $E'=\{ \{a,b\} :N^+(x)=\{a,b\}\subseteq R(V) \text{ for some vertex}\ x\ \text{under illusion} \}$. That is, we get $E'$ by joining the red out-neighbors of the vertices under illusion.
 %    We show that $H$ is acyclic by contradiction. Finally, we recolor a minimum size vertex cover of $H$ to obtain a minimum illusion-free recoloring  of $G$.

We show that $p$-\textsc{Difr}
 is polynomial time solvable not only on out-trees but also when the underlying undirected graph is a tree.
 In search of tractable cases, we further explore structural parameters. Strengthening our results on trees, we obtain fixed-parameter tractable (FPT) algorithm for $p$-\textsc{Difr} parameterized by treewidth of the graph  when the maximum deficiency of any vertex is bounded. We show that $p$-\textsc{Difr} is solvable in time $\mathcal{O}((2 D)^{\textit{tw}})\cdot n^{\mathcal{O}(1)}$  on an $n$-vertex directed graph $G$ with deficiency at most $D$. %Our result can be stated as follows.
 
 Recall that the problem is NP-hard even when $D=1$. To eliminate illusion for a vertex we need to bring down the deficiency to zero. % To prove \Cref{thm:treewidthfpt}, we design a dynamic-programming algorithm, using a nice tree decomposition of the underlying undirected graph, that tracks recoloring decisions for  each bag of the decomposition and the deficiency of each vertex in the bag.
For definitions of treewidth and tree decomposition see \cite{cygan2015parameterized}.
 Finally, we use a simple  ILP based formulation to show that $p$-\textsc{Difr} is FPT parameterized by the number of vertices under illusion, also, paramaterized by the treedepth of the primal as well as the dual LP.

 Our results demonstrate that illusion elimination can be done efficiently when the graphs are sparse; however, even though the flow of information might be more structured in acyclic graphs, it does not help in identifying the vertices to be recolored to eliminate illusion. 
% Due to the paucity of space, all omitted proofs are deferred to the full version of the paper \url{https://arxiv.org/pdf/2604.02395}.

 \subsection{Related Works.}
Different variants of the majority illusion elimination problem have been studied from the graph-theoretic perspective \cite{venema2023graph}.
Grandi et al. and Dipple et al. \cite{grandi2023identifying,dippel2025eliminating} studied the computational questions of eliminating the majority illusion, in \textit{undirected graphs} by editing edges (i.e., adding and/or deleting edges). Grandi et al.~\cite{grandi2023identifying} showed that the problem of eliminating illusion for $q$ fraction of vertices is NP-complete in general and W[1]-hard when parameterized by the number of modified edges (i.e., the solution size). 
Dipple et al.~\cite{dippel2025eliminating} showed that eliminating majority illusion for all agents is solvable in polynomial time. Unlike undirected graphs, eliminating illusions by editing edges in directed graphs is trivial in the sense that it does not utilize the topology of the network and depends only on deficiency of the vertices. 

% The work of Fioravantes et al. \cite{fioravantes2025eliminating} is
% most relevant to us as it addresses the question of eliminating the majority illusion by recoloring the vertices, although in undirected graphs. They show W[2]-hardness by solution size, FPT with respect to treewidth plus the solution size, NP-hardness on planar graphs, and obtain PTAS for planar graphs, analogous to our results. For directed graphs, we show a stronger NP-hardness on grids and explore graph topologies, and parameters to design efficient algorithms.

Conditions for the existence of majority illusion in directed networks and how they differ from the undirected case have been studied \cite{kreutzer2017majority,anastos2021majority,broersma2025extending}.
Majority illusion also affects information diffusion, an extremely well studied problem, where a target set is elected to maximize influence \cite{kempe2005influential}. %The target set selection problem  finds a minimum set of vertices to assign a color to maximize the spread of the given color. The color spreads as vertices get colored based on a threshold of colored vertices in their neighborhood. 
Bredereck and Elkind \cite{bredereck2017manipulating} considered the problem of maximizing opinion diffusion under majority dynamics with the help of bribery (or influence) and established a connection between their model and the target set problem.  Furthermore, controlling the election can be easier by influencing agents' opinion in a social network \cite{faliszewski2022opinion,wu2022manipulating}.
Our problem connects naturally to graph modification problems, where the goal is to identify a minimum size vertex subset such that deleting or modifying it satisfies a desired graph property \cite{cygan2015parameterized}.

\section{Preliminaries}
This section introduces our social network model and establishes the notations used throughout the paper. We also present several key observations that provide the foundation for the subsequent analysis and results.

A social network is represented by a directed graph $G = (V,E)$ where vertices are the agents and each vertex has a labeling (or color) $f : V (G) \to \{B, R\}$, where $B$ ($R$ resp.) denotes blue (red resp.). %As stated earlier, for the \textsc{Difr} problem we assume that there are strictly more vertices with the label $B$ than $R$ in $G$, whereas for $p$-\textsc{Difr} there is no such restrictions.

\textbf{Notation.} We use $[n]$ to denote the set $\{1,2, \dots,n\}$. Given a directed graph $G = (V,E)$, we use $(v,u)$ to denote the edge from $v$ to $u$, for $v, u \in V$.
Let $R(V)=f^{-1}(R)$ and $B(V)=V\setminus R(V)$ denote the set of red and blue vertices, respectively. For any vertex $v\in V$, we denote its sets of out-neighbors and in-neighbors by $N^+(v)=\{u\in V:(v,u)\in E\}$ and $N^-(v)=\{u\in V:(u,v)\in E\}$, respectively.
Let $r_v$ and $b_v$ be the number of red and blue out-neighbors of $v$, respectively, defined as:
$r_v=|N^+(v)\cap R(V)|$ and $b_v=|N^+(v)\cap B(V)|$.\SR{used?}

A vertex $v$ is considered to be \emph{under illusion} if it possesses a surplus of red out-neighbors; that is, if $r_v>b_v$. We define the \emph{deficiency} of a vertex $v$, denoted by \textit{def}$(v)$, as:
\textit{def}$(v)=\text{max}\ \{0, r_v-b_v\}$. %=\text{max}\ \{0, \left\lceil \frac{1}{2}\cdot |N^+(v)|\right\rceil- b_v \}$.
Consequently, the set of illusion vertices, denoted by $X$, consists of all vertices with positive deficiency: $X=\{v\in V:\text{\textit{def}} (v)>0\}$.

The $p$-\textit{deficiency} of a vertex $v$, denoted by \textit{def}$_p(v)$, is defined by \textit{def}$_p(v)=\text{max}\{0, \left\lceil p\cdot |N^+(v)|\right\rceil- b_v\}$. A vertex $v$ is under $p$-\textit{illusion} if \textit{def}$_p(v)>0$. Consequently, a vertex $v$ is $p$-\textit{illusion free} if $b_v\geq \left\lceil p \cdot |N^+(v)| \right\rceil$. Let $X_p$ denote the set of vertices under $p$-illusion, called as the $p$-illusion vertices.

 We say that a \emph{recoloring $f'$ of $f$ produces a solution} to the \textsc{Difr} ($p$-\textsc{Difr}, resp.) problem on $G$ if  $G$ has no illusion ($p$-illusion, resp.) vertex under the coloring $f'$; the solution is the subset of vertices that are recolored, i.e., $\{v \in V(G) : f(v) \neq f'(v)\}$, and solution size $k$ is the size of this set, that is, the number of recolorings.

\subsection{Structural Properties.}
We begin with some basic observations. 
First, we make the trivial observation that it is never of benefit to recolor a blue vertex.
\begin{restatable}{observation}{obsbluerecolor}\label{obs:bluerecolor}
    To eliminate both illusion and $p$-illusion, we always recolor a red vertex and never recolor a blue vertex.
\end{restatable}

\begin{proof}
If a vertex is under illusion or $p$-illusion, then it has a deficit of blue out-neighbors. So to remove illusion or $p$-illusion of the vertex, we need more blue vertex in its out-neighbor, which is possible only by recoloring red out-neighbors to blue.
\end{proof}

\begin{restatable}{observation}{obsrecolordef}\label{obs:recolordef}
    Let $v$ be a vertex under majority illusion in $G$. Then, illusion of $v$ can be eliminated by recoloring at least $t_v=\left\lceil \frac{\textit{def}(v)}{2} \right\rceil$ red out-neighbors of $v$.
\end{restatable}

\begin{proof}
     If we recolor one red vertex in \( N^+(v) \), the red out-neighbors of $v$ changes from $r_v $ to $r_v - 1$ and blue out-neighbors of $v$ changes from $ b_v$ to $b_v + 1$. Thus, modified deficiency of $v$ is
\[updated\_\textit{def}(v) = (r_v - 1) - (b_v + 1) = \textit{def}(v) - 2. \]  
Thus, by recoloring one element in \( N^+(v) \), the deficiency decreases by $2$.
Let $t_v$ denote the number of red vertices of $N^+(v)$ that we need to recolor for $v$ to be illusion-free. 
Then
%\centering
$ \textit{def}(v) - 2 \cdot t_v \leq 0$. 
Therefore, $t_v \geq \frac{\textit{def}(v)}{2} $. Thus, we have 
 $ t_v = \left\lceil \frac{\textit{def}(v)}{2} \right\rceil$.
Therefore, to make $v$ illusion free, we need to look into only $t_v$ elements of $N^+(v)$ and recolor at least \( t_v=\left\lceil \frac{\textit{def}(v)}{2} \right\rceil \) (red) members of \( N^+(v) \).
\end{proof}

    \SR{state it for $p$-illusion?}

 Each recolored vertex in $N^+(v)\cap R(V)$ increases $b_v$ by $1$, hence, decreases \textit{def}$_p(v)$ by $1$. Thus we get the following.
\begin{observation}\label{obs:pdef}
   Let $v$ be a vertex under $p$-illusion in $G$. Then, $p$-illusion of $v$ can be eliminated by recoloring at least \textit{def}$_p(v)$ red out-neighbors of $v$. 
\end{observation}

\section{Hardness Reductions}
 We provide four hardness reductions in this section for \textsc{Difr} and $p$-\textsc{Difr}. Observe that given a recoloring, the non-existence of majority illusion or $p$-illusion can be verified in polynomial time. Thus, our problems are in NP. We prove that \textsc{Difr} remains NP-complete when restricted to grid graphs and $p$-\textsc{Difr} is hard even on bipartite DAGs. Thus, we cannot design polynomial time algorithms for acyclic or planar graphs, unless P=NP. Furthermore, we show that $p$-\textsc{Difr} is NP-complete on directed graphs of degree at most three for $p\in (0,\frac{2}{3}]$. We start with the reduction on grid graphs. 

\subsection{NP-completeness of \textsc{Difr} on Grid Graphs}

We study the \textsc{Difr} problem on planar graphs, especially, on directed grids. Directed grid is a  undirected grid with arbitrary orientations of the edges.
% \SR{this is not done yet. Change such lines.}We have a polynomial time algorithm for outward grids \Cref{prop:outgrid}.\SR{CanNOT start sentences with but, because, and} 
The \textsc{Difr} problem is NP-complete for directed grids. The hardness reduction is from \textsc{Planar Monotone Rectilinear $3$-SAT} problem, a restricted variant of $3$-SAT. In \textsc{Monotone $3$-SAT}, each clause is monotone, that is, it consists exclusively of either positive or negative literals. An \textit{incidence graph} of a $3$-SAT formula  is a bipartite graph. Variables and clauses form the two distinct vertex sets, with an edge connecting a variable to a clause if that variable (as a positive or negative literal) appears in the clause. In a standard \textit{rectilinear embedding} of the incidence graph of \textsc{Planar Monotone Rectilinear 3-SAT} \cite{de2012optimal}, variable vertices are arranged along a central horizontal axis; positive clauses are positioned above this axis, while negative clauses are placed below (see \Cref{fig:gridpicturefinal}(a)).  Each clause is depicted as a horizontal segment connected to its corresponding variables via three vertical segments, referred to as legs.
The clauses can be nested so that the graph is still planar. 

\smallskip
\noindent
\textsc{Planar Monotone Rectilinear $3$-SAT.}\\
\textit{Input.}
A $3$-CNF formula $\mathcal{\varphi}$, such that every clause
of $\mathcal{\varphi}$ is monotone and consists of three literals and the incidence graph of $\mathcal{\varphi}$ has a rectilinear embedding.  \\
\textit{Output.} Does there exist a satisfying assignment for $\mathcal{\varphi}$?

\smallskip
\noindent
This problem is well known to be NP-complete \cite{de2012optimal}. 
We also assume that formula $\mathcal{\varphi}$ consists of distinct clauses.

\begin{restatable}{theorem}{thmgridreduc}\label{thm:gridreduction}
 \textsc{Difr} is NP-complete on directed grids.
\end{restatable}

\begin{proof}
       Let $\mathcal{\varphi}$ be an instance of the \textsc{Planar Monotone Rectilinear 3-SAT} problem. Let $k$ be the number of variables in $\mathcal{\varphi}$. We construct an instance $(\hat{G},f,k')$ of \textsc{Difr}. We will first construct a graph $G$ that is a subgraph of a grid graph. The vertex set of $G$ is constructed  as follows:

    \begin{enumerate}[wide, noitemsep, topsep=0pt, label=(\roman*)]
        \item For every variable $x$, we add two sets of variable vertices, namely the positive variable vertices $\{v^i_x | i \in [\ell_x]\}$ and the negative variable vertices $\{v^i_{\overline{x}}| i \in [\ell_x]\}$, 
        where $\ell_x$ is a positive integer bounded by a polynomial in the size of the formula $\mathcal{\varphi}$ for every $x$. We fix the value of $\ell_x$ when the grid embedding is done.

        \item For every variable $x$, we add another set of vertices $v^{xi}$, called the \textit{controllers} of $x$ where $i \in [2\ell_x]$.
        \item For every clause $c\in \mathcal{\varphi}$, we add a clause-vertex $v_c$ and an additional vertex $v_c^1$.  This is the clause gadget for clause $c$.
    \end{enumerate}

   \smallskip
   \noindent
   The initial color assignment $f$ is as follows: for each variable $x$, all the variable vertices $v^i_x$ and $v^i_{\bar x}$ are colored red, for each $i\in [\ell_x]$, and all the other vertices are colored blue.

    We add the following directed edges to the graph $G$:

   \begin{enumerate}[wide, noitemsep,topsep=0pt,label=(\roman*)]
    \item \begin{sloppypar} For each variable $x$, we construct a gadget by connecting the edges  $(v^{x(2i-1)},v^i_x),(v^{x(2i-1)},v^i_{\overline{x}}),(v^{x(2i)},v^i_{\overline{x}}),(v^{x(2i)},v^{i+1}_{x})$ for each $i\in [\ell_x]$ where the index $\ell_x+1$ is taken modulo $\ell_x$ (see \Cref{fig:gridpath}). Note that the gadget is a cycle in the undirected sense.\end{sloppypar} 

We call this undirected cycle so constructed the \textit{variable gadget} associated with the variable $x$ and is denoted by $V_x$. In this variable gadget, each controller vertex $v^{xi}$ is under illusion. 
    % Picture of the variable gadget 
\begin{figure}[t]
    \centering
    % \includegraphics[width=0.5\linewidth]{}
%   \resizebox{\linewidth}{!}{ 
   \begin{tikzpicture}[
    % Styles for the nodes
    variable_node/.style={circle, draw, minimum size=0.3cm, inner sep=0pt, thick},
    rednode/.style={circle, draw=red, fill=red!5, minimum size=0.3cm, thick},
    bluenode/.style={circle, draw=blue, fill=blue!5, minimum size=0.3cm, thick}, 
    font=\small,
    node distance=0.4cm % Set a standard horizontal gap
]

    % --- 1. First Positive Variable ---
    \node[rednode, label=above:$v^1_x$] (vpx1) {};

    % --- 2. First Controller (Points LEFT to v^1_x and RIGHT to v^1_nx) ---
    \node[bluenode, label=above:$v^{x1}$] (vc1) [right=of vpx1] {};
    \draw[-Stealth] (vc1) -- (vpx1); 

    % --- 3. First Negative Variable ---
    \node[rednode, label=above:$v^1_{\overline{x}}$] (vnx1) [right=of vc1] {};
    \draw[-Stealth] (vc1) -- (vnx1); 

    % --- 4. Second Controller (Points LEFT to v^1_nx and RIGHT to v^2_x) ---
    \node[bluenode, label=above:$v^{x2}$] (vc2) [right=of vnx1] {};
    \draw[-Stealth] (vc2) -- (vnx1); 

    % --- 5. Second Positive Variable ---
    \node[rednode, label=above:$v^2_x$] (vpx2) [right=of vc2] {};
    \draw[-Stealth] (vc2) -- (vpx2); 

    % --- 6. Transition to Dots ---
    \node[bluenode, label=above:$v^{x3}$] (vc3) [right=of vpx2] {};
    \draw[-Stealth] (vc3) -- (vpx2); 

    \node[rednode, label=above:$v^2_{\overline x}$] (vnx2) [right=of vc3] {};
    \draw[-Stealth] (vc3) -- (vnx2);
    
    \node[draw=none] (dots) [right=0.5cm of vnx2] {$\dots$};
    \draw[-Stealth, shorten <=2pt] (dots) -- (vnx2); 

    % --- 7. Final Segment ---
    \node[rednode, label=above:$~~v^{\!\ell_x}_{\!\bar{x}}$] (vlast) [right=1.5cm of dots] {};
    % The last controller
    \node[bluenode, label=92:$v^{x(2\ell_{\!x}\!\!-1\!)}$] (vclast) [left=of vlast] {};
    
    \draw[-Stealth] (vclast) -- (vlast); % Points right to variable
    \draw[-Stealth, shorten <=4pt] (vclast) -- (dots); % Points left to dots

    % --- Closing controller for circular structure ---
    \node[bluenode, label=87:$v^{x(2\ell_x)}$] (vcclose)[right =of vlast] {};

    % Edge to wrap into a cycle
    \draw[-Stealth] (vcclose) -- (vlast);

    % --- Added Low-Profile Curved Feedback Edge ---
    % Matching the thickness, brightness, and arrow type of the straight paths
    \draw[->, >=Stealth, thick] (vcclose) to [out=-155, in=-20, looseness=0.4] (vpx1);
    
\end{tikzpicture}%}
\caption{Variable gadget $V_x$ for variable $x$}
    \label{fig:gridpath}
\end{figure}

        \item \label{itm:clause-variable connection} For every clause-vertex $v_c$, if literal $x$ (or $\overline{x}$) is in $c$, then we add the edge $(v_c,v^j_x)$  (or $(v_c,v^j_{\overline{x}})$ resp.) for a unique $j\in [\ell_x]$.
        % \SR{note that we are not connecting the first and the last vertices on the variable gadget to any clause. Fix the index if it wrong somewhere else as well.}
     The value of $j$ depends on the grid embedding; note that $v^j_x$ or $v^j_{\overline{x}}$ are in the variable gadget associated to the variable $x$.  
        
        \item Add edge $v_c$ to $v_c^1$, for every clause-vertex $v_c$.  In clause gadget, %( the connection of clause vertices to the corresponding variable vertices contained in the clause as in \ref{itm:clause-variable connection}), 
        the clause-vertex $v_c$ is under illusion.
    \end{enumerate}

\begin{figure}[t]
\captionsetup{font=footnotesize}
\centering
% Ensure you have these loaded in your preamble or right before the figure environment:
% \usetikzlibrary{backgrounds, fit, arrows.meta}

% --- LEFT MINIPAGE ---
\begin{minipage}[b]{0.45\columnwidth} 
\centering
% \resizebox{\linewidth}{!}{% 
\begin{tikzpicture}[
    box/.style={draw, rectangle, minimum width=0.3cm, minimum height=0.3cm},
    clause/.style={circle, fill=black, inner sep=1.5pt},
    font=\small,
    line/.style={thick}
]
\node[box] (x1) at (0,0) {$x_1$};
\node[box] (x2) at (1,0) {$x$};
\node[box] (x3) at (2,0) {$x_3$};
\node[box] (x4) at (3,0) {$x_4$};
\node[clause,label=above:{$c_1$}] (c1) at (1,0.8) {};
\node[clause,label=below:{$c_2$}] (c2) at (2,-0.8) {};
\node[clause,label=below:{$c_3$}] (c3) at (1,-1.3) {};
\draw[line,thin] (x1.east) -- (x2.west); 
\draw[line,thin] (x2.east) -- (x3.west); 
\draw[line,thin] (x3.east) -- (x4.west);
\draw[line] (x1.north) -- (0,0.8) -- (3,0.8) -- (x4.north); 
\draw[line] (1,0.8) -- (x2.north);
\draw[line] (x2.south east) -- (1.21,-0.8) -- (2.998,-0.8) -- (x4.south); 
\draw[line] (x3.south) -- (2,-0.8);
\draw[line] (x1.south) -- (0,-1.28) -- (3.27,-1.28) -- (x4.south east); 
\draw[line] (x2.south) -- (1,-1.28);
\end{tikzpicture}%}
\caption*{(a) Rectilinear embedding: variable vertex $x$ belongs to a positive clasuse $c_1$ and two negative clauses $c_2,\ c_3$.}
\end{minipage}% 
\hfill 
% --- RIGHT MINIPAGE ---
\begin{minipage}[b]{0.48\columnwidth}
\centering
% \resizebox{\linewidth}{!}{% 
\begin{tikzpicture}[
    rednode/.style={circle, draw=red!70!black, fill=red!30, minimum size=0.15cm, inner sep=0pt, thick},
    bluenode/.style={circle, draw=blue!70, fill=blue!15, minimum size=0.15cm, inner sep=0pt, thick},
    edge/.style={->, thin}, >=Stealth,
    edgee/.style= {draw=black!17,thin},
    gadgetbox/.style={draw=green!60!black, fill=green!10, dashed, rounded corners, inner sep=2pt}
]
\node[bluenode, label=right: $v_{c_1}$] (vc1)  at (2,1.5) {};
\node[rednode, label=right:$v_x^1$]  (vx1p) at (2,1) {};
\draw[edge] (vc1) -- (vx1p);

\node[bluenode,label=right:{$v^{x1}$}] (b1) at (2,0.5) {};
\draw[edge] (b1) -- (vx1p);

\node[rednode, label=right:$v_{\overline{x}}^1$]  (vx2p) at (2,0) {};
\draw[edge] (b1) -- (vx2p);
\node[bluenode] (b2) at (2,-0.5) {};
\draw[edge] (b2) -- (vx2p);

\node[rednode, label=above:$v_{\overline{x}}^j$] (vxjp) at (3.2,-0.5) {};
\draw[dashed] (b2) -- (vxjp);
\node[bluenode, label=below: {$v_{c_2}$}] (vc2) at (3.7,-0.5) {};
\draw[edge] (vc2)--(vxjp);
\node[bluenode] (b3) at (3.2,-1.2) {};
\draw[dashed] (b3)--(vxjp);
\node[rednode] (r1) at (2,-1.2) {};
\draw[dashed] (b3)--(r1);
\node[bluenode] (b4) at (1.5,-1.2) {};
\node[rednode, label=60: $v^q_{\overline{x}}$] (vxpq) at (1,-1.2) {};
\draw[edge] (b4)--(vxpq);
\draw[edge] (b4)--(r1);
\node[bluenode,label=below:$v_{c_3}$] (vc3) at (1,-1.7) {};
\draw[edge] (vc3)--(vxpq);

\node[bluenode] (b7) at (1,-0.5) {};
\draw[dashed] (vxpq)--(b7);

\node[rednode] (r2) at (1,0) {};
\draw[edge] (b7)--(r2);
\node[bluenode] (b8) at (1,0.5) {};
\draw[edge] (b8)--(r2);
\node[rednode, label={[label distance=-4pt]left:{\footnotesize $v^{\ell _x}_{\overline{x}}$}}] (vxlx) at (1, 1) {};
\draw[edge] (b8)--(vxlx);

\node[bluenode,label={[label distance=0pt]above:{\footnotesize$v^{x(2\ell _x\!)}$}}] (vlast) at (1.5,1) {};
\draw[edge] (vlast) -- (vxlx); 
\draw[edge] (vlast) -- (vx1p);
% \draw[dashed] (vxpq)--(vlast);

\begin{scope}[on background layer]
    \node[gadgetbox, fit=(vx1p) (vx2p) (b1) (b3) (b8) (r1) (r2)  (b2) (vxjp) (vxpq) (b4) (vxlx) (vlast)] {};
    \draw[edgee](0.7,1.5)--(4,1.5);
    \draw[edgee](0.7,1)--(4,1);
    \draw[edgee](0.7,0.5)--(4,0.5);
    \draw[edgee](0.7,0)--(4,0);
    \draw[edgee](0.7,-0.5)--(4,-0.5);
    \draw[edgee](0.7,-1.2)--(4,-1.2);
    \draw[edgee](0.7,-1.7)--(4,-1.7);
    \draw[edgee](1,1.6)--(1,-2);
    \draw[edgee](1.5,1.6)--(1.5,-2);
    \draw[edgee](2,1.6)--(2,-2);
    \draw[edgee](3.2,1.6)--(3.2,-2);
    \draw[edgee](3.7,1.6)--(3.7,-2);
\end{scope}

\end{tikzpicture}%}
\caption*{(b) Grid embedding of gadget $V_x$ for variable $x$ (in the left figure). The dashed lines indicate subpaths of the gadget.}
\end{minipage}
\vspace{0.3cm}
\caption{Grid reduction excerpt from \Cref{thm:gridreduction}.}
\label{fig:gridpicturefinal}
\end{figure}

We show that the constructed graph $G$ is planar by providing the following embedding.

\noindent
\textbf{Embedding of $G$.} To find a grid embedding of the graph constructed so far, we adapt the planar rectilinear embedding of the incidence graph $G_{\varphi}$ of $\mathcal{\varphi}$. We embed the variable gadget for a variable $x$ along the horizontal and vertical legs connected to $x$ in  the rectilinear embedding (see \Cref{fig:gridpicturefinal}). Recall that each clause is monotone, i.e., it is either positive or negative clause.

\smallskip
\noindent
\textbf{Embedding the clause gadget.} Let $c$ be a clause of size three in formula $\mathcal{\varphi}$. In the embedding of $G$, we position the clause vertex $v_c$ at the intersection of the clause's horizontal segment and its central vertical leg in the rectilinear embedding. For example, in \Cref{fig:gridpicturefinal}(b) vertex $v_{c_1}$ is placed at a coordinate corresponding to the position marked as $c_1$ in \Cref{fig:gridpicturefinal}(a).
For positive clauses, an auxiliary vertex $v^1_c$ is placed above $v_c$ (not shown in the figure), while for negative clauses, it is placed below. 

\smallskip
\noindent
\textbf{Embedding the variable gadget.} We embed the variable gadget so that it maintains the connection to its clause vertices. 
The variable gadget $V_x$ for each variable $x$ is routed along the existing connections in the rectilinear embedding to reach each of its clauses. In particular, if the clause is connected to the variable by the central leg, then the variable gadget reaches the same column as the clause vertex and the row immediately below (resp. above) the clause vertex for a positive (resp. a negative) clause. On the other hand, if the clause is connected to the variable by a leg that is left  (resp. right) to the variable, then the variable gadget reaches the same row as the clause vertex and the column immediately left (resp. right) to the clause vertex.   %To accommodate this, we introduce auxiliary horizontal and vertical lines parallel to the original rectilinear segments.
For example, suppose that variable $x$ appears in one positive clause $c_1$ and two negative clauses $c_2$ and $c_3$ (Fig.~\ref{fig:gridpicturefinal}(a)).  The construction proceeds as follows (see Fig.~\ref{fig:gridpicturefinal}(b)):

   % \smallskip
   % \noindent
   \textit{Connecting the Gadget for $x$ to the clauses containing $x$:} Without loss of generality, suppose that the vertex $v^1_x$ is placed at the same column as $v_{c_1}$ and the row immediately below $v_{c_1}$. Then, the positive literal vertex $v_x^1$ is connected to the positive clause vertex $v_{c_1}$, with the variable gadget following the vertical leg from clause $c_1$ to the variable's horizontal axis.
   
 The path continues along the vertical and horizontal segments to reach the row of the vertex $v_{c_2}$. Observe that $c_2$ is connected to $x$ by its left leg. Thus, the path reaches the column immediately left to the column of $v_{c_2}$.  A directed edge $(v_{c_2}, v^j_{\overline{x}})$ is added from the clause vertex $v_{c_2}$ to the corresponding negative variable literal $v^j_{\overline{x}}$.

 The gadget then traverses back along the horizontal lines to reach the column of the vertex $v_{c_3}$ as $x$ and $c_3$ are in the same column in the rectilinear embedding (in \Cref{fig:gridpicturefinal}(a)). We add the edge $(v_{c_3}, v^q_{\overline{x}})$ for some index $q\in [\ell _x]$.

    \textit{Completing the variable gadget of $x$}: Finally, the variable gadget returns via the vertical leg of $c_3$, running parallel to its initial trajectory along the leg of $c_1$. The cycle is closed by placing a controller vertex $v^{x(2\ell_x)}$ between $v^{\ell _x}_{\overline{x}}$ and $v^1_x$, with a directed edge to $v^1_x$, completing the cycle of the variable gadget of $ x$ (see the highlighted part in Fig.~\ref{fig:gridpicturefinal}(b)).

Let us denote the graph constructed using the clause-variable connections as defined above by $G$.

\textbf{Properties of $G$.} In the construction of $G$ and $f$ we maintain the following properties. We use these properties to prove structural lemmas for $G$. 
\begin{enumerate}[label= (P\arabic*)]
\item\label{it:gridreddeficiency} Each clause vertex and each controller vertex is under majority illusion with deficiency one.
    \item\label{it:gridrednointersection1} By construction, the variable gadget for each variable forms a simple cycle (in undirected sense) along the grid lines.
    \item Each clause vertex is connected to a red vertex corresponding to each literal that appear in it.
    \item\label{it:gridrednointersection2} The subpath of a variable gadget that connects two of its clauses is always on odd number of vertices. Moreover, the distance between the clause vertices are large enough to avoid any overlap or intersection between two paths from variable gadgets of two distinct variables. A distance bounded by polynomial in the size of $\varphi$ is sufficient to maintain this property.
    \item\label{it:gridredpolylx} Given an embedding of $G$, we define $\ell_x$ as the number of negative variable vertices i.e., $v^i_{\overline x}$ (and equivalently, the number of positive variable vertices i.e., $v^i_x$) on the variable gadget for variable $x$. Thus, $\ell_x$ is bounded by a polynomial in size of $\varphi$.
    \item\label{it:gridredpropv} Each two consecutive columns (and rows) start with different colors, i.e, if one starts with red, the next one starts with blue.  
\end{enumerate}

 We begin by showing that we are close to constructing a grid. 

    \begin{lemma}\label{lem:gridminorporp}
        The graph $G$ is a subgraph of a grid; specifically, $G$ can be obtained by deleting a set of vertices and edges from a grid graph. Moreover, in $G$ there are no two consecutive red vertices in a row or in a column.
        % Moreover, in $G$ there are no two consecutive rows or columns that have two consecutive red vertices. 
    \end{lemma}

    \begin{proof}
   
By the above construction, the graph $G$ is embedded into the plane using only orthogonal segments (horizontal and vertical lines). Since all vertices are positioned at the intersections of a discrete rectangular coordinate system and no two edges cross except at these vertices, by \ref{it:gridrednointersection1} and \ref{it:gridrednointersection2}, the graph $G$ is isomorphic to a subgraph of a sufficiently large rectangular grid. 
% It follows immediately that $G$ is a subgraph of a grid.

    Now we prove the second part of the statement of the lemma. According to the construction of the variable gadget, every red variable vertex $v^i_x$ or $v^i_{\overline{x}}$ is preceded and followed by a blue controller vertex $v^{xj}$ for each $i \in [\ell_x], j \in [2\ell_x]$. Since we are not allowed to cross two variable gadgets, it is impossible for two red vertices to be next to each other. Additionally, by \ref{it:gridredpropv} of the construction, we know that no two columns (or rows) have red vertex in the same row (resp. column). This completes the proof of Lemma~\ref{lem:gridminorporp}.
 \end{proof}

Therefore, $G$ is a subgraph of a grid. We now construct a grid graph $\hat{G}$ from $G$. 

\noindent
    \textbf{Construction of the grid $\hat{G}$}: In the graph $G$, it may be possible that some vertices as well as some edges are missing. 
To complete the grid, we add a dummy blue vertex (say, $d$) at an empty position and set $d$ as the out-neighbor of the vertices around it in the grid. Due to second part of \Cref{lem:gridminorporp}, we can add each missing edge with blue vertex as the head (that is, receives the edge). Thus, we ensure that no dummy vertex is under illusion.
    
 We set $k'=
     \sum_{x \in Var(\varphi)} \ell_x$ where $Var(\varphi)$ denotes the set of variables in the formula $\varphi$. %The sum is over total number of variables in the formula $\mathcal{\varphi}$. 
     Then $k'$ is polynomially bounded by $|\mathcal{\varphi}|$, the size of the formula $\mathcal{\varphi}$ using \ref{it:gridredpolylx}. So, we have the required constructions for the reduction.  Thus, the construction can be done in polynomial time. Let us now prove the equivalence of the \textsc{Difr} problem and the \textsc{Planar Monotone Rectilinear 3-SAT} problem.

  Note that the clause vertices $v_c$ (with three red and one blue out-neighbors), and the controller vertices  $v^{xi}$ (two red out-neighbors) for $i \in [2\ell_x]$ are under illusion, for every variable $x$. We claim that, $\mathcal{\varphi}$ is satisfiable if and only if there exists a illusion-free recoloring for $\hat{G}$ of size at most $k'$. 

\textit{Forward Direction.} If $\varphi$ is satisfiable, there exists a truth assignment $t$ from the variables to \{\textsc{True, False}\}. For every variable $x$, for each $i\in [\ell_x]$, we set 
      \begin{center}
      $f'(v^i_x)=B, f'(v^i_{\overline{x}})=R$ if $t(x)=$\textsc{True},\\ \vspace{2mm}
      $f'(v^i_x)=R, f'(v^i_{\overline{x}})=B$ if $t(x)=$\textsc{False}.
      \end{center} 

      For all other vertices $v$ in $\hat{G}$, we set $f'(v)=f(v)$. In this way, $f$ and $f'$ differ on precisely $k'$ vertices. For every variable $x$ and every $i\in [2\ell_x]$, the controller vertices $v^{xi}$ in the variable gadget associated to $x$ are not under illusion in $f'$ since the label of precisely all the $v^i_x$ or all the $v^i_{\overline{x}}$ was changed. For every clause $c$, we know that $c$ was satisfied by $t$. Hence, there is at least one variable-vertex (positive or negative) that in-neighbors $v_c$ whose label was changed and $v_c$ is not under illusion in $f'$. These were the only vertices under illusion in $f$, thus $\hat{G}$ is illusion-free under the coloring function $f'$.

 \textit{Backward Direction.} Let us assume that there exists an illusion-free recoloring $f'$ for $\hat{G}$ of size $k'$. Since no vertex is under illusion under $f'$, we have the following lemma:

 \begin{lemma}\label{lem:gridback}
      \begin{enumerate}[label=(\roman*)]
     \item\label{lem:gridback_1} For every variable $x$, the function $f'$ recolors exactly $\ell_x$ vertices in the variable gadget for $x$.
      \item\label{lem:gridback_2} The function $f'$ recolors either all the positive variable vertices ($v^i_x$'s), or all the negative variable vertices ($v^i_{\overline{x}}$'s) in the variable gadget $V_x$ for $x$.
      \end{enumerate}
 \end{lemma}
\begin{proof}

     We first prove \ref{lem:gridback_1}. The variable gadget for $x$ has $2\ell_x$ red vertices and $2\ell_x$ blue vertices. Note that at most two blue vertices become illusion-free by one recoloring. So,  we need at least $\left\lceil\frac{2\ell_x}{2}\right\rceil=\ell_x$ recolorings for $2\ell_x$
     blue vertices. Furthermore, if $f'$ recolors more than $\ell_x$ vertices for some $x$, then the total number of recolorings is greater than 
         $\sum_{x}\ell_x=k'$
     , a contradiction to the fact that the solution size is $k'$.

     Now we prove \ref{lem:gridback_2}. Suppose both positive and negative variable vertices are recolored in a variable gadget for some variable $x$. Then we will show a contradiction to \Cref{lem:gridback}.\ref{lem:gridback_1}. 
     Observe that there exist a positive and a negative variable vertex that are not recolored as $\ell_x$ vertices are recolored from the variable gadget and there are $\ell_x$ positive and $\ell_x$ negative vertices in the cycle. 
     Let $v^s_x$ and $v^t_{\overline{x}}$ be the two nearest positive and negative variable vertices, respectively, that are not recolored and there are $\ell$ positive variable vertices in this subpath between $v^s_x$ and $v^t_{\overline{x}}$, including the vertices $v^s_x$ and $v^t_{\overline{x}}$, in the variable gadget for $x$.
     Let us call this subpath $P_{s,{\overline{t}}}$. We have the following observation.

     \begin{observation}\label{obs:grid_subpath}
      \begin{enumerate}[label=(\roman*)]
          \item\label{it:ell1} Number of vertices on the path $P_{s,{\overline{t}}}$ is strictly more than one. 
          \item\label{it:tworecolored} The vertices $v^{s-1}_{\overline{x}}$ and $v^{t+1}_x$ are recolored.
          \item\label{it:subpathlength} total number of vertices recolored in the path $P_{s,{\overline{t}}}$ is $2\ell-2$.
      \end{enumerate} 
     \end{observation}
     \begin{proof}

        (i)That is, first, we prove that $\ell>1$, otherwise the unique controller vertex between $v^s_x$ and $v^t_{{\overline{x}}}$ is under illusion, a contradiction.

       (ii) In path $P_{s,{\overline{t}}}$, every positive and negative variable vertex other than $v^s_x$ and $v^t_{\overline{x}}$ have been recolored by $f'$, as $v^s_x$ and $v^t_{\overline{x}}$ are the two nearest vertices that are not recolored. Consider the in-neighbors of $v^s_x$ and $v^t_{{\overline{x}}}$ that are outside the subpath $P_{s,{\overline{t}}}$ and let these in-neighbors be $v^{xs}$ and $v^{xt}$, respectively (see \Cref{fig:explanation obs. 4}). 
     Since we have started with an illusion-free recoloring, and $v^s_x$ and $v^t_{\overline{x}}$ are not recolored, so the other out-neighbors of $v^{xs}$ and $v^{xt}$, namely, $v^{s-1}_{\overline{x}}$ and $v^{t+1}_x$, respectively, have been recolored.

\begin{figure}[t]
\captionsetup{font=footnotesize}
    \centering
    \begin{tikzpicture}[
        variable_node/.style={circle, draw, minimum size=0.3cm, inner sep=0pt, thick},
        rednode/.style={circle, draw=red, fill=red!5, minimum size=0.3cm, thick},
        bluenode/.style={circle, draw=blue, fill=blue!5, minimum size=0.3cm, thick},
        font=\small,
        node distance=0.4cm 
    ]

        % --- Nodes ---
        \node[rednode, label=above:$v^1_x$] (vpx1) {};

        \node[draw=none] (dots1) [right=0.5cm of vpx1] {$\dots$};
        
        \draw[-Stealth, shorten <=2pt] (dots1)-- (vpx1) ;

        \node[bluenode, label=above:{$v^{s-1}_{\overline{x}}$}] (vpxs-1) [right= of dots1 ] {};
        \draw[-Stealth] (dots1) -- (vpxs-1);

        \node[bluenode, label=above:{$v^{xs}$}] (vcs) [right=of vpxs-1] {};
        \draw[-Stealth] (vcs)--(vpxs-1);

        \node[rednode, label=above:{$v^{s}_x$}] (vpxs) [right= of vcs ] {};
        \draw[-Stealth] (vcs) -- (vpxs);

        \node[draw=none] (dots2) [right=0.5cm of vpxs] {$\dots$};

        \node[rednode, label=above:{$v^{t}_{\overline{x}}$}] (vpxt) [right= of dots2 ] {};

        \node[bluenode, label=above:{$v^{xt}$}] (vct) [right=of vpxt] {};
        \draw[-Stealth] (vct)--(vpxt);

        \node[bluenode, label=above:{$v^{t+1}_x$}] (vpxt+1) [right= of vct ] {};
        \draw[-Stealth] (vct) -- (vpxt+1);

        \node[draw=none] (dots3) [right=0.5cm of vpxt+1] {$\dots$};
        \draw[-Stealth] (dots3) -- (vpxt+1);

        \node[bluenode, label=above:{$v^{x(2\ell_x)}$}] [right=of dots3](vclast) {};
        \draw[-Stealth] (vclast)--(dots3);

        % --- Fixed Feedback Edge ---
        % Uses -Stealth tip, thick weight, and native black color to match the layout arrows
        \draw[-Stealth, thick] (vclast) to [out=-160, in=-20, looseness=0.4] (vpx1);
        
        % --- Compact Horizontal Curly Brace ---
        \draw[decorate, decoration={calligraphic brace, mirror, amplitude=4pt}, thick] 
            ([yshift=-0.45cm]vpxs.west) -- ([yshift=-0.45cm]vpxt.east)
            node[pos=0.5, yshift=-0.35cm] {$P_{s,\overline{t}}$};
        \draw[decorate, decoration={calligraphic brace, amplitude=4pt}, thick] 
            ([yshift=0.8cm]vpxs-1.west) -- ([yshift=0.8cm]vpxt+1.east)
            node[pos=0.5, yshift=0.35cm] {$P_{\overline{s-1},t+1}$};
        
        \end{tikzpicture}
    \caption{Explanation of Observation~\ref{obs:grid_subpath}: after recoloring $v^{s-1}_{\overline{x}}$ and $v^{t+1}_x$ become blue; subpath $P_{s,\overline{t}}$ has $\ell$ positive variables with $2\ell-2$ recoloring; subpath $P_{\overline{s-1},t+1}$ has $2\ell$ recoloring}
    \label{fig:explanation obs. 4}
\end{figure}

     (iii) Since $v^s_x$ and $v^t_{\overline{x}}$ are the nearest non-recolored vertices, so the number of vertices recolored in the subpath $P_{s,{\overline{t}}}$ is $2\ell-2$, as there are $\ell$ positive variable vertices in this subpath between $v^s_x$ and $v^t_{\overline{x}}$, including the end vertices. Moreover, in the subpath $P_{\overline{s-1},t+1}$ (the subpath from $v^{s-1}_{\overline{x}}$ to $v^{t+1}_x$), number of vertices recolored is $(2\ell-2)+2=2\ell$. There are equal number of positive and negative variable vertices in the subpatha $P_{\overline{s-1},t+1}$ and that number is $\ell+1$. This completes the proof of the observation. 
\end{proof}

To get a contradiction, we need the total number of recolorings on the cycle. Thus, we next count the number of recolored vertices in the remaining variable gadget.

First, observe that since the controller vertex $v^{x(2i)}$ is illusion-free, one of the two vertices $v^i_x$ and $v^{i}_{\overline{x}}$ is recolored.
     Let us consider the remaining part of the variable gadget $V_x$, that is, the subpaths: path $P_{1,s-1}$ (the path from $v^1_x$ to $v^{s-1}_x$) and path $P_{\overline{t+1},{2{\ell_x}}}$ (the path from $v^{t+1}_{\overline{x}}$ to $v^{x2\ell_x}$). Then there are $\ell_x-(\ell+1)$ positive variable vertices in these two subpaths. Since we have an illusion-free recoloring, so at least $\ell_x-(\ell+1)$ variable vertices have been recolored in the two subpaths. %combining both the subpaths $P_{1,s-1}$ and $P_{\overline{t+1},{\overline{\ell_x}}}$. 

     Finally,  total number of recolored vertices in the variable gadget $V_x$ is $2+ (2\ell -2) + \ell_x-(\ell+1)$ as from Observation ~\ref{obs:grid_subpath}\ref{it:tworecolored} and \ref{it:subpathlength}, we have that the two vertices $v^{s-1}_{\overline{x}}$ and $v^{t+1}_x$ are recolored,  $2\ell-2$ vertices are recolored in the path $P_{s,{\overline{t}}}$, and $V_x= P_{1,s-1} \cup \{v^{s-1}_{\overline{x}}\}\cup P_{s,{\overline{t}}} \cup \{v^{t+1}_x\} \cup P_{\overline{t+1},2\ell_x}$. 
     Consequently, $\ell_x+\ell-1$ vertices are colored that is strictly more than $\ell_x$ as $\ell>1$ (Observation ~\ref{obs:grid_subpath}\ref{it:ell1}), a contradiction to the fact that $f'$ colored $\ell_x$ vertices in the path (\Cref{lem:gridback}.\ref{lem:gridback_1}).
This completes the proof of \ref{lem:gridback_2}, the second part of Lemma~\ref{lem:gridback}.
 \end{proof}
 
   We construct the truth assignment by setting $x$ to \textsc{True} (\textsc{False}, resp.) for each variable $x$ such that $f'(v^i_x)=B$ ($f'(v^i_{\overline{x}})=B$, resp.) for all $i\in [\ell_x]$. This assignment is well defined due to \Cref{lem:gridback}\ref{lem:gridback_2} and it satisfies  $\mathcal{\varphi}$.
  
   Finally, we show that it yields a satisfying truth assignment for $\mathcal{\varphi}$. Since no clause vertex remains under illusion, the recoloring function $f'$ have recolored at least one variable vertex adjacent to each clause vertex $v_c$. This implies that every clause contains at least one literal set as \textsc{True}, confirming that $\mathcal{\varphi}$ is satisfiable. This completes the proof.
\end{proof}

    We observe that \Cref{thm:gridreduction} additionally shows that the undirected version of the problem (as studied in \cite{fioravantes2025eliminating}) is NP-hard for subgraphs of grids. However, it is not clear if the undirected version of the problem remains hard on grid graphs. This is because in the construction of \Cref{thm:gridreduction}, the graph $G$ is a subgraph of grid and removing the orientation of the edges in $G$ does not change if a vertex is under illusion. However, the orientation of edges in the graph $\hat{G}$ is crucial to keep the red vertices illusion free and the controller vertices under illusion. 

\begin{corollary}
Given an undirected graph $G$ and a coloring function $f: V(G) \rightarrow \{B,R\}$, deciding if recoloring $k$ vertices eliminates majority illusion is NP-hard even when $G$ is a subgraph of a grid.
\end{corollary}

\subsection{Hardness of p-\textsc{Difr} on Bipartite DAGs}

Now we study the computational difficulty of the $p$-\textsc{Difr} problem. We present a reduction from the classical NP-complete problem \textsc{Hitting Set}. 
In $p$-\textsc{Difr}, each vertex under $p$-illusion requires a certain fraction of its out-neighbors be recolored blue in order to eliminate the $p$-illusion. These requirements can be viewed as demanding that at least a prescribed number of vertices be chosen from specific neighborhoods, precisely mirroring the hitting requirement in set systems.
We show this by providing  a polynomial time reduction from  \textsc{Hitting Set} to $p$-\textsc{Difr} that preserves the parameter, solution size.

\noindent
\textsc{Hitting Set}.\\
\textbf{Input:} A universe $U$, a family $\mathcal{A}$ of sets over $U$, and a nonnegative integer $k$. \\
\textbf{Question:} Does there exist a subset $U' \subseteq U$ of size at most $k$ that has a nonempty intersection with all elements of $\mathcal{A}$?

It is known that \textsc{Hitting Set} is NP-complete \cite{garey1990computers} and W[2]-hard when parameterized by solution size \cite{cygan2015parameterized}. We constraint the minimum size of the sets  in the family by $1/p$. Note that this version of \textsc{Hitting Set} remains W[2]-hard as the sizes of the sets are unbounded.
 
\begin{restatable}{theorem}{thmpdifrreduction}\label{thm:pdifrreduction}\label{thm:W[2]hard}
    $p$-\textsc{Difr} is NP-complete and W[2]-hard for any $p\in (0,1)$ even on acyclic bipartite graphs. 
\end{restatable}

\begin{proof}

We show this by providing  a polynomial time reduction from  \textsc{Hitting Set} to $p$-\textsc{Difr} that preserves the parameter, solution size.   Let $U=\{u_1, u_2, \ldots, u_n\}$ be the universe and $\mathcal{A}=\{A_1, A_2, \ldots, A_m\}$ be the family of sets of sizes $d_1,d_2, \dots,d_m$ where $\min_{j \in [m]} d_j> \left\lceil \frac{1}{p}\right\rceil$. Let $(U, A, k)$ be an instance of the $d$-\textsc{Hitting Set}. We will now construct the graph for the $p$-\textsc{Difr} problem (see \Cref{fig:difrreduction}). 

\begin{enumerate}[noitemsep,label=(\roman*)]
    \item For each element $u_i$, create a red vertex $r_i$ and for each set $A_j$, create a blue vertex $v_j$. Create $x$ dummy blue vertices $b_1, b_2, \ldots, b_x$ where $x = \left\lceil \frac{pd}{1-p}\right\rceil -1$ and $ d = \max_{j \in [m]} d_j$.
\item  For each $u_i\in A_j$, add edge $(v_j, r_i)$, and add edges $(v_j, b _{\ell})$, for $\ell=1, 2, \ldots, x_j$, for each $j\in[m]$, where $x_j$ is a natural number such that $\left\lfloor \frac{pd_j-1}{1-p}\right\rfloor \leq x_j<\left\lceil \frac{pd_j}{1-p}\right\rceil$. Note that the number $x_j$ exists since $p<1$ and $d_j> \left\lceil \frac{1}{p}\right\rceil$ (i.e., $pd_j-1>0$).
\end{enumerate}
        Let us call the graph so created as $G=(V, E)$. This completes the construction. It can be done in polynomial time.

        \begin{figure}[t]
  \centering
\begin{tikzpicture}[
    % Styles for the nodes
    variable_node/.style={circle, draw, minimum size=0.3cm, inner sep=0pt, thick},
    rednode/.style={circle, draw=red, fill=red!5, minimum size=0.3cm, thick},
    bluenode/.style={circle, draw=blue, fill=blue!5, minimum size=0.3cm, thick}, 
    font=\small,
    node distance=0.4cm % Set a standard horizontal gap
]

% --- Bottom row (v_1 ... v_m)
\node[bluenode, label={[label distance=-0.2pt]120:{$v_1$}}] (v1) {};
\node[bluenode, right=of v1, label={$v_2$}] (v2) {};
\node[bluenode, right=of v2, label={[label distance=-0.2pt]20:{$v_3$}}] (v3) {};
\node[right=0.1cm of v3] (vdots) {$\dots$};
\node[bluenode, right=0.2cm of vdots, label={$v_m$}] (vm) {};

% --- Top row (u_1 ... u_n)
% Reduced from 1.2cm to 0.6cm to shorten the vertical length
\node[rednode, above=0.6cm of v1, label={$r_1$}] (u1) {}; 
\node[rednode, right=of u1, label={$r_2$}] (u2) {};
\node[rednode, right=of u2, label={$r_3$}] (u3) {};
\node[right=0.1cm of u3] (udots) {$\dots$};
\node[rednode, right=0.2cm of udots, label={$r_n$}] (un) {};

% --- Rightmost nodes (p_1 ... p_r)
\node[bluenode, below=0.6cm of v1, label={[label distance=-0.2pt]120:{$b_1$}}] (p1) {};
\node[bluenode, below=0.6cm of v2, label={[label distance=-0.5pt]20:{$b_2$}}] (p2) {};
\node[bluenode, below=0.6cm of v3, label={[label distance=-0.3pt]20:{$b_3$}}] (p3) {};
\node[right=0.1cm of p3] (pdots) {$\dots$};
\node[bluenode, below=0.6cm of vm, label={$b_x$}] (pr) {};

% --- Edges
\draw[-stealth] (v1) -- (u1);
\draw[-stealth] (v1) -- (u2);
\draw[-stealth] (v1) -- (p1);
\draw[-stealth] (v2) -- (u1);
\draw[-stealth] (v2) -- (u3);
\draw[-stealth] (v2) -- (un);
\draw[-stealth] (v2) -- (p1);
\draw[-stealth] (v2) -- (p2);

\end{tikzpicture}
\caption{Construction in \Cref{thm:pdifrreduction}, with  $A_1=\{u_1,u_2\}$, $A_2=\{u_1,u_3,u_n\}$}
\label{fig:difrreduction}
\end{figure}

    In our construction, each blue vertex $v_j$ is characterized by $d_j$ red out-neighbors and $x_j$ blue out-neighbors. Our objective is to determine a value for $x_j$ such that the $p$-deficiency of every $v_j$ is exactly $1$. Under this condition, the recoloring of a single red out-neighbor is sufficient to eliminate the $p$-illusion for all such vertices.
    
    \begin{observation} Let $v_j$ for $j \in [m]$ denote a blue vertex in $G$.
        \begin{enumerate}[wide, noitemsep,topsep=0pt,label=(\roman*)]
            \item Vertex $v_j$ is under $p$-illusion.
            \item Recoloring of a single red out-neighbor of $v_j$ eliminates its $p$-illusion.
        \end{enumerate}
    \end{observation}
    \begin{proof}[Proof of the above Observation]
        Each blue vertex $v_j$, for $j \in [m]$, initially has $d_j$ red out-neighbors and the total out-neighbors of $v_j$ is $|N^+(v_j)|= x_j+d_j$.  Since $x_j<\left\lceil \frac{pd_j}{1-p}\right\rceil$, it can be checked that $v_j$ is under $p$-illusion, we have the first statement.

        Furthermore, since $x_j \geq \left\lfloor \frac{pd_j-1}{1-p}\right\rfloor$, a simple calculation shows that $x_j+1\geq p(x_j+d_j)$. As recoloring one red vertex increases the number of blue vertices to $x_j+1$, we have the second statement.
    \end{proof}

    In light of the above observation, we go into the rest of the proof. Let $(G, f, k)$ be the constructed instance of the $p$-\textsc{Difr} problem.  We claim the following
    \textit{$(U, \mathcal{A}, k)$ is a yes-instance for $d$-\textsc{Hitting Set} if and only if $(G, f, k)$ is a yes-instance for $p$-\textsc{Difr}.}

    Suppose $(U, \mathcal{A}, k)$ is a yes-instance of the \textsc{Hitting Set}. Let $X\subseteq U$ with $|X|\leq k$ be such that $X\cap A_j\neq\emptyset$ for all $A_j\in \mathcal{A}$. We recolor the red vertices $R_X$ correspond to $X$ in $G$. Note that $|R_X|=|X|\leq k$. By the construction of the graph we need to recolor exactly one red out-neighbor for every vertex under $p$-illusion(which are the blue vertex) 
    % of each blue vertex 
    to remove its $p$-illusion. Since every $A_j\in \mathcal{A}$ has non empty intersection with $X$, so every blue vertex $v_j$ has at least one red out-neighbor which has been recolored. Hence, $(G, f, k)$ is a yes-instance for $p$-\textsc{Difr}.

    Conversely, Suppose $(G, f, k)$ is a yes-instance for $p$-\textsc{Difr}. Let $R\subseteq R(V)$ be a solution for $p$-\textsc{Difr} problem. Consider the subset $X\subseteq U$ that corresponds to set $R$. Then $|X|\leq k$. Note that $X\cap A_j\neq \emptyset$ for each $A_j\in \mathcal{A}$, because $X\cap A_j=\emptyset$ implies that $v_j$ is under $p$-illusion by the construction of the graph. Hence, $(U, \mathcal{A}, k)$ is a yes-instance for the \textsc{Hitting Set}.
\end{proof}

\subsection{p-\textsc{Difr} for Degree Three Directed Graphs}

In \Cref{thm:gridreduction}, we showed that \textsc{Difr} is NP-hard for graphs of degree at most four. We strengthen the result and  show that $p$-\textsc{Difr} is NP-complete for $p\in (0,\frac{1}{3}]\cup (\frac{1}{3},\frac{2}{3}]$, and is solvable in polynomial time for $p\in (\frac{2}{3},1)$, when the degree of each vertex in the graph is at most three. Note that, for $p\in (\frac{2}{3},1)$, we need every out-neighbor of a vertex to be blue in order to be $p$-illusion-free, so $p$-\textsc{Difr} is trivially polynomial time solvable in this case. We prove the NP-completeness  for $p\in (\frac{1}{3},\frac{2}{3}]$ by a reduction from the \textsc{$3$-Regular Vertex Cover} Problem~\cite{garey1990computers}.

We provide a reduction from a $3$-regular \textsc{Vertex Cover} instance $H$, where in $G$, we create a red vertex for each vertex in $H$ and a blue edge-vertex for each edge in $H$. The out-neighbors of an edge-vertex are its endpoints in $H$ and  it has a personal blue out-neighbor. We show that a subset of $k$ vertices cover all the edges in $H$ if and only if recoloring the corresponding red vertices in $G$ eliminates the $p$-illusion of the edge-vertices for $p \in (1/3,2/3]$.

\smallskip
\noindent
\textsc{$3$-Regular Vertex Cover}:

\textit{Input}: A $3$-regular undirected graph  $H=(U,F)$ and a nonnegative integer $k$.

\textit{Question}: Does there exists a subset $C \subseteq U$ with $|C|\le k$ such that every edge has at least one end point in $C$?

\begin{restatable}{theorem}{thmdegthreereduction}\label{thm:deg3reduction}
 $p$-\textsc{Difr} is NP-complete on graphs of degree at most $3$ for $p\in (\frac{1}{3},\frac{2}{3}]$.
\end{restatable}
\begin{proof}
    We provide a polynomial time reduction from \textsc{$3$-Regular Vertex Cover} to $p$-\textsc{Difr}. Let $(H,k)$ be an instance of \textsc{$3$-Regular Vertex Cover}. We will construct an instance $(G,f,k)$ of $p$-\textsc{Difr} (see \Cref{fig:deg3reduction}). 
    The vertex set of $G$ consists of the following vertices with the initial coloring $f$ as assigned:
\begin{enumerate}[noitemsep,label=(\roman*)]
    \item For every vertex $u\in U$, create a vertex $x_u$ and color it red.
    \item For every edge $e\in F$, create two vertices $x_e,b_e$ and color them blue. 
\end{enumerate}
The edge set of $G$ consists of the following edges:
\begin{enumerate}[noitemsep,label=(\roman*)]
    \item For every edge $e=\{u,v\}\in F$, create the directed edges $(x_e,x_u)$ and $(x_e,x_v)$.
    \item Create the directed edge $(x_e,b_e)$, for every $e\in F$.
\end{enumerate}
This completes the construction of $G$ for an instance of $p$-\textsc{Difr}. Then, note that $x_e$ is under $p$-illusion for every $e\in F$, with two red and one blue  out-neighbors. To remove the $p$-illusion we need to recolor at least one red out-neighbor to blue for every $p\in (\frac{1}{3},\frac{2}{3}]$.

Let $(G, f, k)$ be the constructed instance of the $p$-\textsc{Difr} problem.  We claim the following
    \textit{$(H, k)$ is a yes-instance for \textsc{$3$-Regular Vertex Cover} if and only if $(G, f, k)$ is a yes-instance for $p$-\textsc{Difr}.}

    Suppose $(H,k)$ is a yes-instance for \textsc{$3$-Regular Vertex Cover}. Let $C\subseteq U$ be a vertex cover of size at most $k$. Let us consider the set $V'=\{x_u\in V:u\in C\}$. Then note that $|V'|\le k$ and recoloring each element of $V'$ to blue turns $G$ into a $p$-illusion-free graph, because every edge in $H$ has an end point in $C$. 

    Conversely, suppose $(G,f,k)$ is a yes-instance for $p$-\textsc{Difr}. Let $R'\subseteq R(V)$ be a solution for $p$-\textsc{Difr} of size at most $k$. Then, consider the set $U'=\{u\in U:x_u\in R'\}$. Then $U'$ is a vertex cover of $H$ of size at most $k$ as every illusion vertex $x_e$ for $e=\{u,v\}\in F$ has at least one out-neighbor $x_u$ (or $x_v$) that is recolored. This completes the proof of the theorem. 

\begin{figure}[t]
  \centering
\begin{tikzpicture}[
    % Styles for the nodes
    variable_node/.style={circle, draw, minimum size=0.3cm, inner sep=0pt, thick},
    rednode/.style={circle, draw=red, fill=red!5, minimum size=0.3cm, thick},
    bluenode/.style={circle, draw=blue, fill=blue!5, minimum size=0.3cm, thick}, 
    font=\small,
    node distance=0.4cm % Set a standard horizontal gap
]

% --- Bottom row (v_1 ... v_m)
\node[bluenode, label={[label distance=-0.2pt]120:{$x_e$}}] (xe) {};
\node[bluenode, right=of xe, label={[label distance=-0.2pt]-10:{$x_g$}}] (xg) {};
\node[right=0.1cm of xg] (xdots) {$\dots$};
\node[bluenode, right= 0.2cm of xdots, label={[label distance=-0.2pt]20:{$x_h$}}] (xh) {};

% \node[bluenode, right=0.2cm of vdots, label={$v_m$}] (vm) {};

% --- Top row (u_1 ... u_n)
% Reduced from 1.2cm to 0.6cm to shorten the vertical length
\node[rednode, above=0.6cm of xe, label={$x_u$}] (xu) {}; 
\node[rednode, right=of xu, label={$x_v$}] (xv) {};
\node[right=0.1cm of xv] (vdots) {$\dots$};
% \node[rednode, right=of u2, label={$r_3$}] (u3) {};

\node[rednode, right=0.2cm of vdots, label={$x_w$}] (xw) {};

% --- Rightmost nodes (p_1 ... p_r)
\node[bluenode, below=0.6cm of xe, label={[label distance=-0.2pt]120:{$b_e$}}] (be) {};
\node[bluenode, below=0.6cm of xg, label={[label distance=-0.5pt]20:{$b_g$}}] (bg) {};
% \node[bluenode, below=0.6cm of v3, label={[label distance=-0.3pt]20:{$b_3$}}] (p3) {};
\node[right=0.1cm of bg] (bdots) {$\dots$};
\node[bluenode, below=0.6cm of xh, label={[label distance=-0.2pt]20:{$b_h$}}] (bh) {};

% --- Edges
\draw[-stealth] (xe) -- (xu);
\draw[-stealth] (xe) -- (xv);
\draw[-stealth] (xe) -- (be);
\draw[-stealth] (xg) -- (xv);
\draw[-stealth] (xg) -- (xw);
\draw[-stealth] (xg) -- (bg);
\draw[-stealth] (xh) -- (vdots);
\draw[-stealth] (xh) -- (xw);
\draw[-stealth] (xh) -- (bh);

\end{tikzpicture}
\caption{Construction in \Cref{thm:deg3reduction}, two edges   $e=\{u,v\}$, $f=\{v,w\}$, and $h=\{w,z\}$ in $F=E(H)$, corresponding $x_z$ is not shown}
\label{fig:deg3reduction}
\end{figure}
\end{proof}
It remains to show that it is NP-complete when $p\in (0,\frac{1}{3}]$. Note that for $p\in (0,\frac{1}{3}]$, a vertex under $p$-illusion needs one blue out-neighbor to be $p$-illusion-free, that is, a vertex is under $p$-illusion if all its out-neighbors are red. We start from a $3$-\textsc{Hitting Set} instance and modify the reduction in \Cref{thm:pdifrreduction} by omitting the vertices $\{b_1, \dots, b_x\}$ in the construction  to obtain the following result.

\begin{restatable}{proposition}{degthreereductionforpzerotoonethird}\label{prop:deg3reductionfor p zero to one-third}
 $p$-\textsc{Difr} is NP-complete on graphs of degree at most $3$ for $p\in (0,\frac{1}{3}]$.
\end{restatable}
\begin{proof}
    We show this by providing  a polynomial time reduction from  $3$-\textsc{Hitting Set} to $p$-\textsc{Difr}.   Let $U=\{u_1, u_2, \ldots, u_n\}$ be the universe and $\mathcal{A}=\{A_1, A_2, \ldots, A_m\}$ be the family of sets of sizes at most $3$. We construct an instance $(G,f,k)$ of the $p$-\textsc{Difr}.

    The vertex set of $G$ and the initial colors to them is as follow (see \Cref{fig:deg3reductionfor p zero to one-third}): for every $u_i\in U$, create a red vertex $r_i$, and for every set $A_j\in \mathcal{A}$, create a blue vertex $v_j$.

    The edge set consists of the following directed edges: Create the edge $(v_j,r_i)$ if $u_i\in A_j$. 

    \begin{figure}[t]
  \centering
\begin{tikzpicture}[
    % Styles for the nodes
    variable_node/.style={circle, draw, minimum size=0.3cm, inner sep=0pt, thick},
    rednode/.style={circle, draw=red, fill=red!5, minimum size=0.3cm, thick},
    bluenode/.style={circle, draw=blue, fill=blue!5, minimum size=0.3cm, thick}, 
    font=\small,
    node distance=0.4cm % Set a standard horizontal gap
]

% --- Bottom row (v_1 ... v_m)
\node[bluenode, label={[label distance=-0.2pt]120:{$v_1$}}] (v1) {};
\node[bluenode, right=of v1, label={$v_2$}] (v2) {};
\node[bluenode, right=of v2, label={[label distance=-0.2pt]20:{$v_3$}}] (v3) {};
\node[right=0.1cm of v3] (vdots) {$\dots$};
\node[bluenode, right=0.2cm of vdots, label={$v_m$}] (vm) {};

% --- Top row (u_1 ... u_n)
% Reduced from 1.2cm to 0.6cm to shorten the vertical length
\node[rednode, above=0.6cm of v1, label={$r_1$}] (u1) {}; 
\node[rednode, right=of u1, label={$r_2$}] (u2) {};
\node[rednode, right=of u2, label={$r_3$}] (u3) {};
\node[right=0.1cm of u3] (udots) {$\dots$};
\node[rednode, right=0.2cm of udots, label={$r_n$}] (un) {};

% --- Edges
\draw[-stealth] (v1) -- (u1);
\draw[-stealth] (v1) -- (u2);
\draw[-stealth] (v1) -- (u3);
\draw[-stealth] (v2) -- (u1);
\draw[-stealth] (v2) -- (u3);
\draw[-stealth] (v2) -- (un);
% \draw[-stealth] (v2) -- (p1);
% \draw[-stealth] (v2) -- (p2);

\end{tikzpicture}
\caption{Construction in \Cref{prop:deg3reductionfor p zero to one-third}, with  $A_1=\{u_1,u_2,u_3\}$, $A_2=\{u_1,u_3,u_n\}$}
\label{fig:deg3reductionfor p zero to one-third}
\end{figure}    
This completes the construction of $G$. Note that every vertex has degree at most three and the blue vertices are under illusion with at most three red out-neighbors. Since $p\in (0,\frac{1}{3}]$, so we need one blue out-neighbor for each blue vertex to be $p$-illusion-free. We claim the following
    \textit{$(U, \mathcal{A}, k)$ is a yes-instance for $3$-\textsc{Hitting Set} if and only if $(G, f, k)$ is a yes-instance for $p$-\textsc{Difr}.}

    Suppose \textit{$(U, \mathcal{A}, k)$} is a yes-instance for $3$-\textsc{Hitting Set}. Let $U'\subseteq U$ be a solution with $|U'|\le k$. Consider the set $V'=\{r_i\in R(G):u_i\in U'\}$. Then, $V'$ is a solution of $p$-\textsc{Difr} with size at most $k$, as $U'$ hits $A_j$ for every $j$, so we have at least one out-neighbor of $v_j$ that is recolored.

    Conversely, suppose $(G, f, k)$ is a yes-instance for $p$-\textsc{Difr}. Let $V'\subseteq R(G)$ be a solution of size at most $k$. Let us consider the set $U'=\{u_i\in U:r_i\in V'\}$. Then note that $U'$ hits every set $A_j$, as every $v_j$ is $p$-illusion-free, so $v_j$ has at least one out-neighbor which is recolored to blue and that out-neighbor is in $V'$. This completes the proof.
    \end{proof}

\section{Polynomial time Solvable Topologies}
Although the problem is NP-complete on DAGs, we present a dichotomy by showing that the problem is polynomial time solvable on trees. Moreover, we identify orientations of grid where the problem can be solved in polynomial time even though it is hard on grids.
%In this section, we present polynomial-time algorithms for several graph classes.
We begin by establishing a basic observation for directed cycles, showing that eliminating $p$-illusions in these structures is straightforward. We then extend our analysis to cycles and outward grids before concluding with a specialized algorithm for trees.

\begin{restatable}{observation}{propdrcycle}\label{obs:drcycle}
    The $p$-\textsc{Difr} problem on directed cycles can be solved by recoloring all red vertices.
\end{restatable}

Before we explain the algorithm for cycles, we state a simple reduction rule that will be used in multiple algorithms.

  \begin{rr}\label{rr:outdegone}\label{rr:grid}
        If a vertex in the graph $G$ has only one out-neighbor $v$ that is red, then recolor $v$.
    \end{rr}

    \Cref{rr:outdegone} is \textit{safe}, as the vertex has only one red out-neighbor, so we need to recolor it to remove the $p$-illusion of the vertex.

We extend the above proposition to consider a directed graph whose underlying undirected version is a cycle, and then consider \textit{outward grids}, a directed grid such that the out-neighbors of a vertex $v$ are the vertices to the right of $v$ and below $v$. Notably, in both graph classes, every vertex $v$ satisfies $|N^+(v)|\leq 2$. This inherent sparsity simplifies the elimination of illusion and $p$-illusion, leading to the following observation.

\begin{restatable}{observation}{obspgreaterhalf}\label{obs:pgreaterhalf}
    Suppose $G=(V, E)$ is a directed graph with $|N^+(v)|\leq 2$ for every $v\in V$. Then,  
\begin{enumerate}[label=(\roman*)]
    \item For $p>\frac{1}{2}$, the problem becomes trivial, as resolving all illusions merely requires recoloring every red vertex except for the red vertices that have no in-neighbors.
    
    \item For $p\leq \frac{1}{2}$, the notions of illusion and $p$-illusion of a vertex $v$ are equivalent, i.e., a solution for \textsc{Difr} yields a solution for $p$-\textsc{Difr} and vice-versa.
\end{enumerate}
\end{restatable}

\begin{proof}
    \begin{enumerate}[label=(\roman*)]
    \item We have $p>\frac{1}{2}$. note that $|N^+(v)|=1$ implies that $b_v\geq\left\lceil p\cdot |N^+(v)|\right \rceil=1$ and $|N^+(v)|=2$ implies that $b_v\geq \left\lceil p\cdot |N^+(v)|\right \rceil=2$. So in both cases, we have $b_v=|N^+(v)|$. Thus, to eliminate $p$-illusion, every red vertex with an in-neighbor must be recolored.

    \item Here $p\leq \frac{1}{2}$. We will show that a vertex $v$ is under illusion if and only if it is under $p$-illusion.

    A vertex $v$ that is under illusion implies that either of the following is true. Either the vertex has exactly one out-neighbor and that is red, or has two red out-neighbors. Furthermore, as $p\leq \frac{1}{2}$, these two conditions are both necessary and sufficient for the vertex $v$ to be under $p$-illusion. In both the cases, to eliminate illusion and $p$-illusion, we need to recolor exactly one red out-neighbor, and the same recolored out-neighbor will work for illusion and $p$-illusion. 
\end{enumerate}
\end{proof}

Our next two results are for eliminating illusion in cycles and outward grids. We will see that in both of these structures the number of out-neighbors are bounded by two.
Thus, due to \Cref{obs:pgreaterhalf}, in \Cref{prop:cycle,prop:outgrid} for cycles and outward grids, respectively, we assume $p \leq 1/2$ and solve for the \textsc{Difr} problem as \textsc{Difr} and $p$-\textsc{Difr} are same in this set up. % We design an algorithm that eliminates illusion by recoloring the minimum number of vertices. 

\begin{restatable}{proposition}{propcycle}\label{prop:cycle}
    The \textsc{Difr} problem on an $n$-vertex graph whose underlying undirected graph is a cycle can be solved in time $\mathcal{O}(n)$.
\end{restatable}
\begin{proof}
     Let $G$ be an $n$-vertex directed graph such that the underlying undirected graph is a cycle. 
    First note that a vertex can have either at most two out-neighbors, or one in-neighbor and one out-neighbor, or at most two in-neighbors.
     
    Illusion occurs in the first two cases as  a vertex with two in-neighbors and no out-neighbors is never under illusion.
    We divide the set of vertices under illusion in two categories.
    Let $X_1$ denote the set of vertices under illusion that have exactly one out-neighbor that is red.
    Let $X_2$ denote the set of vertices under illusion that have two red out-neighbors. The algorithm runs the following steps:

    \paragraph{Algorithm.}
     
 \begin{enumerate}[leftmargin= 12mm, label=Step \arabic*]
        %\item\label{it:dircyclstep1} Apply \Cref{rr:outdegone} exhaustively.
        \item\label{it:dircyclstep1} Apply \Cref{rr:outdegone} exhaustively.
        \item\label{it:dircyclstep2} While there is a vertex $v$ that has two red out-neighbors $r_1$ and $r_2$.
        \begin{enumerate}[leftmargin= 0mm,label=(\alph*)]
            \item If $r_1$ and $r_2$ have no other in-neighbor under illusion other than $v$, then recolor any one of $r_1$ and $r_2$.
            \item If $r_1$ and $r_2$ both have illusion in-neighbors other than $v$, then recolor exactly one of $r_1$ and $r_2$.
            \item  Finally, let $r \in \{r_1,r_2\}$ such that $r$ has an in-neighbor under illusion other than $v$, then recolor $r$.
            \item If applicable, then execute \Cref{rr:outdegone} exhaustively.
        \end{enumerate}
    \end{enumerate}

    Note that in both the cases \ref{it:dircyclstep2}(a) and \ref{it:dircyclstep2}(b), it does not matter which red vertex ($r_1$ or $r_2$) is recolored. We write the two cases separately for the ease of showing the correctness.
   
    \textit{Correctness.} In \ref{it:dircyclstep1}, applying \Cref{rr:outdegone} exhaustively, recolors $|X_1|$ red vertices. Its correctness follows from the safeness of \Cref{rr:outdegone}. Consequently, any minimum recoloring need to recolor $|X_1|$ red vertices. Next we show that \ref{it:dircyclstep2} eliminates illusion by recoloring minimum number of vertices.

In \ref{it:dircyclstep2}, each vertex under illusion is from the set $X_2$, i.e., has two red out-neighbors.
First observe that \ref{it:dircyclstep2} considers all three possible neighborhoods of a vertex in $X_2$ and eliminates its illusion. Thus, after \ref{it:dircyclstep2}, no vertex of $X_2$ is under illusion. We show that \ref{it:dircyclstep2} recolors minimum number of vertices. Each vertex $v$ in $X_2$ must have at least one red out-neighbor recolored. We use this property to show that in the cases (a) - (c), we always recolor optimally.

If the algorithm recolors in \ref{it:dircyclstep2}(a), then since the out-neighbors of vertex $v$ are not out-neighbors of any other illusion vertex and we need to color at least one of them to eliminate $v$'s illusion, thus, it recolors optimally. If we recolor using \ref{it:dircyclstep2}(b), then since both out-neighbors of $v$ are also out-neighbors of another illusion vertex, recoloring a vertex $r\in \{r_1,r_2\}$, eliminates illusion of two vertices, namely $v$  and the other in-neighbor of $r$. Since recoloring a vertex can eliminate the illusion of at most two vertices, we recolor optimally in case (b). The argument to show  optimality of case (c) is the same as case (b). Thus, the algorithm always optimally recolors to eliminate illusion of vertices in $X_2$.
     
    The steps run in polynomial-time. We need $\mathcal{O}(|X_1+X_2|)=\mathcal{O}(n)$ time to eliminate the illusion of the graph.
\end{proof}

Although, the problem is NP complete on grids (\Cref{thm:gridreduction}), we show that not all orientations of grid poses the hardness. We identify a class of oriented grids, called the  \emph{outward grids} where the problem is solvable in polynomial time.
% , We next design a  polynomial-time algorithm for
Recall that an \emph{outward grid} is a directed grid such that it has a planar embedding where the out-neighbors of a vertex $v$ are the vertices to the right of $v$ and below  $v$. 

\begin{restatable}{proposition}{propoutgrid}\label{prop:outgrid}
    Let $G=(V,E)$ be an $m\times n$ outward grid. Then, \textsc{Difr} on $G$  can be solved in time $\mathcal{O}(mn\sqrt{mn})$.
\end{restatable}
\begin{proof}
    Given $G$ is an outward grid. Recall that an outward gird is defined as a directed grid such that it has an embedding where the out-neighbors of a vertex $v$ are the vertices to the right of $v$ and below  $v$. The deficiency of a vertex under illusion is either $1$ or $2$ since each vertex has at most two in-neighbors and at most two out-neighbors. Observe that a vertex with deficiency $1$ has exactly one out-neighbor. 
    
\paragraph{Algorithm.}
\begin{enumerate}[leftmargin= 12mm, label=Step \arabic*]

\item Eliminate illusion of the vertices with deficiency $1$ using \Cref{rr:outdegone}.
\item Construct an auxiliary undirected graph  $H=(R(V),E')$, where $R(V)$ is the set of red vertices, and
    $E'=\{ \{a,b\} :N^+(x)=\{a,b\} \text{ for some } x \in X\}$.
\item Find a minimum vertex cover of $H$.
\item Recolor the minimum vertex cover.
\end{enumerate}
    \textit{Correctness.} We first eliminate illusion of the vertices with deficiency $1$ using \Cref{rr:outdegone}, which states that if a vertex in the graph $G$ has only one out-neighbor $v$ that is red, then recolor $v$. 
    
    After exhaustively apply \Cref{rr:grid},
    % \SR{use  ref{}. No hard coded reference}
     each vertex of $G$ that is under illusion has deficiency $2$. We remove their illusion in the following way. First, we construct an auxiliary graph $H$. For every $x\in X$ (set of illusion vertices), we have that $N^+(x)=\{a,b\}$, for some $a,b \in R(V)$, that is, the two out-neighbors of $x$ are red. We construct an undirected graph $H=(R(V),E')$,
    where $R(V)$ is the set of red vertices, and
    $E'=\{ \{a,b\} :N^+(x)=\{a,b\} \text{ for some } x \in X\}$. That is, we get $E'$ by connecting the two red out-neighbors for each illusion vertex.
    In Step 4 of the algorithm we recolor a subset of red vertices that forms a vertex cover of $H$. We show that it removes illusion for each vertex of $X$, and does so with the minimum number of recolorings. We need to recolor exactly one red out-neighbor for every vertex under illusion since to remove illusion we must recolor at least $\lceil \frac{def(x)}{2}\rceil$ red out-neighbors for every illusion vertex $x$ (\Cref{obs:recolordef}). Therefore, by choosing a minimum vertex cover we are guaranteed that we recolor minimum number of vertices.
    % \SR{Need to argue why this is minimum.}

Finally, we show that the vertex cover of $H$ can be found in polynomial time by finding a maximum matching in $H$ as $H$ is a bipartite graph.
\begin{claim}\label{clm:Gbipartite}
    $H$ \textit{is acyclic}.
\end{claim}

 \begin{proof}[Proof of \Cref{clm:Gbipartite}] We show that $H$ has no  cycle. Observe that no cycle of length two is possible. First we show that $H$ has no $3$-cycle $(a_1,b_1,a_2)$. Thus, we show if $N^+(x_1)=\{a_1,b_1\}$, $N^+(x_2)=\{a_2,b_1\}$ for some $x_1,x_2\in X$ and $a_1,a_2,b_1\in R(V)$, then there exist no $x_3\in X$ such that $N^+(x_3)=\{a_1,a_2\}$. We show it by contradiction.

    Suppose that a vertex $x_3$ with the above property exists in $X$. Then we have the two scenarios depicted in \Cref{fig:outgrid}.

    For existence of the vertex $x_3\in X$ with neighborhood $N^+(x_3)=\{a_1,a_2\}$, we need $a_1$ and $a_2$ to be in two consecutive level, which is not possible. This proves that there is no $3$-cycles in $H$. Similarly, we can show $H$ has no cycle. Thus, we prove the claim and complete the proof of the proposition.
    \end{proof}

       \begin{figure}[b]
    \centering
    % First diagram
    \begin{minipage}[b]{0.45\linewidth}
        \centering
        \begin{tikzpicture}[
    % Styles for the nodes
    variable_node/.style={circle, draw, minimum size=0.3cm, inner sep=0pt, thick},
    rednode/.style={circle, draw=red, fill=red!5, minimum size=0.3cm, thick},
    graynode/.style={circle, draw=blue, fill=blue!5, minimum size=0.3cm, thick}, 
    font=\small,
    node distance=0.4cm % Set a standard horizontal gap
]
            % Nodes
            \node[graynode, label=left:{$x_1$}] (x1) at (0.8,0.8) {};
            \node[rednode, label=right:{$a_1$}] (a1) at (0.8,0) {};
            \node[graynode, label=left:{$x_2$}] (x2) at (1.6,1.6) {};
            \node[rednode, label=right:{$b_1$}] (b1) at (1.6,0.8) {};
            \node[graynode, label=left:{$x_3$}] (x3) at (0,0) {};
            \node[rednode, label=right:{$a_2$}] (a2) at (2.4,1.6) {};

            % Solid arrows
            \draw[->, thick] (x1) -- (a1);
            \draw[->, thick] (x1) -- (b1);
            \draw[->, thick] (x2) -- (b1);
            \draw[->, thick] (x2) -- (a2);
            \draw[->, thick] (x3) -- (a1);

            % Dashed arrows
            % \draw[dashed, thick] (x1) -- (b1);
            \draw[->, dashed, thick] (x3) -- (a2);
        \end{tikzpicture}
        \caption*{First scenario}
    \end{minipage}
    \hfill
    % Second diagram
    \begin{minipage}[b]{0.45\linewidth}
        \centering
       \begin{tikzpicture}[
    % Styles for the nodes
    variable_node/.style={circle, draw, minimum size=0.3cm, inner sep=0pt, thick},
    rednode/.style={circle, draw=red, fill=red!5, minimum size=0.3cm, thick},
    graynode/.style={circle, draw=blue, fill=blue!5, minimum size=0.3cm, thick}, 
    font=\small,
    node distance=0.4cm % Set a standard horizontal gap
]
            % Nodes
            \node[graynode, label=left:{$x_1$}] (x1) at (5,1.6) {};
            \node[rednode, label=right:{$a_1$}] (a1) at (5.8,1.6) {};
            \node[graynode, label=left:{$x_2$}] (x2) at (4.2,0.8) {};
            \node[rednode, label=right:{$b_1$}] (b1) at (5,0.8) {};
            \node[graynode, label=left:{$x_3$}] (x3) at (3.4,0) {};
            \node[rednode, label=right:{$a_2$}] (a2) at (4.2,0) {};

            % Solid arrows
            \draw[->, thick] (x1) -- (b1);
            \draw[->, thick] (x2) -- (b1);
            \draw[->, thick] (x3) -- (a2);
            \draw[->, thick] (x2) -- (a2);
            \draw[->, thick] (x1) -- (a1);

            % Dashed arrows
            \draw[->, dashed, thick] (x3) -- (a1);
            % \draw[dashed, thick] (x2) -- (a1);
        \end{tikzpicture}
        \caption*{Second scenario}
    \end{minipage}
    \caption{The two cases in the proof of \Cref{prop:outgrid}}
    \label{fig:outgrid}
\end{figure}

    \noindent
    \textit{Running time.} We Construct $H$ in polynomial time. Then, we find a maximum matching in $H$ using Hopcroft-Karp \cite{hopcroft1973n}. Finally, find a minimum vertex cover from the maximum matching (using K\"{o}nig's theorem) and recolor the red vertices in the cover. Thus, the algorithm runs in polynomial time. 
    For illusion vertices with deficiency one, we need $\mathcal{O}(mn)$ time.
    For illusion vertices with deficiency two, we need $\mathcal{O}(|E'|\sqrt{|R|})=\mathcal{O}(mn\sqrt{mn})$ time. 
    Hence, the total run time is 
    $\mathcal{O}(mn+mn\sqrt{mn}) =\mathcal{O}(mn\sqrt{mn})$.
    %, which is polynomial in input size.
    \end{proof}

 We focus on directed trees. For out-trees (all edges going away from root and from parent to children), we simply need to recolor minimum number of red children for each vertex that is under $p$-illusion. The problem is more involved when we consider arbitrary directed trees. We devise a dynamic programming algorithm. We consider an arbitrary fixed ordering of the vertices in $T$, in which in-children of a vertex precedes its out-children.
 
\begin{restatable}{theorem}{thmdt}\label{thm:dt}
       The $p$-\textsc{Difr} problem on an n-vertex directed tree can be solved in time $\mathcal{O}(n^4)$.
   \end{restatable}

   \begin{proof}
    Let $T=(V,E)$ be a directed rooted tree and $r$ be the root of the tree $T$. We begin with a simple rule similar to \Cref{rr:outdegone}.
\begin{restatable}{rr}{rrleafreduc}\label{rr:leafillusion}
 If a leaf node is under $p$-illusion, recolor its parent. 
 \end{restatable}
   
The safeness of \Cref{rr:leafillusion} follows from the fact that a leaf node under $p$-illusion implies that its parent is red, so we have no choice but to recolor its parent to eliminate its $p$-illusion. 
 
  After applying \Cref{rr:leafillusion} exhaustively, we now have $p$-illusion-free leaf nodes.

 We will devise a dynamic programming algorithm for trees. 

 \paragraph{Algorithm.}
  Let $\mathcal{\sigma}$ be an arbitrary fixed ordering of the vertices of $T$, in which in-children of a vertex precedes its out-children. Let $u \in V(T)$ have $\ell$ in-children, $v_1, v_2, \dots,  v_\ell$ and $s$ out-children, $w_1, w_2, \dots, w_s$. Let $T_u$ denote the set of vertices in the subtree rooted at $u$ including $u$. We abuse the notation and write $T_u \setminus \{u\}$ to denote the forest we get from $T_u$ after deleting the vertex $u$.

 We define 
     $C[u,z_u,flag_u,w_q]$ to be the minimum number of vertices recolored in $T_u$ such that no vertex in $T_u\setminus \{u,w_{q+1},... ,w_{s}\}$ has $p$-illusion where $\{w_{q+1},... ,w_{s}\}$ are the out-children of $u$ that succeed $w_q$ in the ordering $\mathcal{\sigma}$,
    and 
     \begin{align*}
         u &= \text{current root},\\
         z_u &= \text{number of out-children of $u$ that are recolored in $T_u$},\\
         flag_u &= \begin{cases}
             0 & \text{if $u$ is not recolored},\\
                1 & \text{if $u$ is recolored}
         \end{cases}\\
        w_q &= \text{$q^{th}$ out-child of $u$ in the ordering $\sigma$}.
     \end{align*}
    Then, we define the following recursive formulas:

    \begin{align}\label{eq:inchi}
     C[u,z_u,flag_u,\mathcal{\phi}]&  = \sum_{\substack{i:\\v_i\in N\!^{-}\!(u)\cap X}} \min_{z:z+flag_u \geq \text{\textit{def}}_p(v_i) }C[v_i,z,flag_{v_i},last(v_i)]
    \end{align}
    and   
    
\begin{align}\label{eq:outchi}
      C[u,z_u,flag_u,w_q]&  =\min \begin{cases} 
        \displaystyle\min_{z: w_q\in X\ and\ z\geq\text{\textit{def}}_p(w_q) } C[w_q,z,0,last(w_q)] +~C[u,z_u,flag_u,w_{q-1}] , \\
         \displaystyle\min_{z: w_q\in X\ and\ z\geq\text{\textit{def}}_p(w_q)} C[w_q,z,1,last(w_q)] +~C[u,z_u-1,flag_u,w_{q-1}]
        \end{cases} 
    \end{align}

    In \Cref{eq:inchi} and \Cref{eq:outchi}, $last(x)$ indicates the last out-child of the vertex $x$.

    Note that \Cref{eq:inchi} processes all the in-children of the vertex $u$ and \Cref{eq:outchi} processes out-children of the vertex one by one. Note that when computing $C[u,z_u,flag_u,w_1]$, we use \Cref{eq:inchi} in \Cref{eq:outchi} to evaluate the value of $C[u,\allowbreak z_u,\allowbreak flag_u,\allowbreak w_{q-1}]$ and $C[u,\allowbreak z_u-1,\allowbreak flag_u,\allowbreak w_{q-1}]$ using the convention that $w_{1-1}=w_0= \mathcal{\phi}$. We compute the entries in a bottom-up approach, and  the final goal is to compute $C[r,z_r,flag_r,last(r)]$ with $z_r\geq \text{\textit{def}}_p(r)$. We show the correctness of the algorithm in the next lemma.

   \begin{lemma}
        Our algorithm for directed trees gives a minimum recoloring at each step.    
   \end{lemma}
   
 \begin{proof} 
       We will proceed by induction on the number of vertices of the tree that has been processed. At any vertex $u$, we show that $C[u,z_u,flag_u,w_q]$ is the minimum number of vertices recolored in $T_u$ such that no vertex in $T_u\setminus \{u,w_{q+1},\dots ,w_{s}\}$ has $p$-illusion where $\{w_{q+1}, \dots ,w_{s}\}$ are the out-children of $u$ that succeed $w_q$ in the ordering $\mathcal{\sigma}$.
       % We show that no vertex in $T_u\setminus \{u,w_{q+1},\dots ,w_{s}\}$ has illusion after recoloring  $C[u,z_u,flag_u,w_q]$ red vertices  where $\{w_{q+1},... ,w_{s}\}$ are the out-child of $u$ that succeed $w_q$ in the ordering $\mathcal{\sigma}$ and this is a minimum recoloring.
       
    \textit{Base Case.} Suppose that the leaf vertices has been processed. For the leaf vertices that are in-children of a vertex, by \Cref{rr:leafillusion}, we already have a minimum recoloring to remove its $p$-illusion. Note that no leaf vertex that is an out-child of a vertex can be under $p$-illusion, as it has no out-neighbor. Thus, we have a minimum recoloring eliminating the $p$-illusion of the leaf vertices.

\begin{sloppypar}

\textit{Inductive Step.} We show the correctness of the recursive step using the two cases for \Cref{eq:inchi,eq:outchi}.
\begin{itemize}[label=, leftmargin=0pt] 
    \item \textit{Case 1:} Computing the entry $C[u,z_u,flag_u,\emptyset]$, i.e., we are processing in-children of $u$.

    \begin{enumerate}[label=Case 1.\arabic*:,align=left]
     \item  $flag_u=1$.
     
 That is, $u$ is recolored.
     Note that on the right-hand side of \Cref{eq:inchi} $z$ denotes the number of out-children of $v_i$ (in-child of $u$) that are recolored. Thus, total $z+1$ out-neighbors of $v_i$ is recolored. Moreover, in recurrence, the entry $C[u,z_u,flag_u,\mathcal{\phi}]$ only considers the minimum value of $z$ in the summation such that at most $z+1$ vertices need to be recolored for $v_i$ to be $p$-illusion-free. 

    Therefore, we recolor the minimum number of vertices in $T_u$ to remove the $p$-illusion of in-children of $u$.

    \item  $flag_u=0$.

That is $u$ is not recolored.
    Then in order for $v_i$ to be $p$-illusion free (if ever possible by recoloring only the out-children of $v_i$), at least $\text{\textit{def}}_p(v_i)$ out-children of $v_i$ must be recolored. Now, similar to Case 1.1, we can say, that we are doing a minimum  recolorings of the out-children of $v_i$ in \Cref{eq:inchi} and recolor at least $\text{\textit{def}}_p(v_i)$ many to remove $p$-illusion of in-children of $u$. %in $T_u\setminus \{u,w_{1},... ,w_{s}\}$.
    \end{enumerate}
    Now, by the induction hypothesis, the entries $C[v_i,z_u,flag_u,last(v_i)]$ for all $v_i$ that is a in-child of $u$ have been correctly computed. Thus, recoloring $C[u,z_u,flag_u,\emptyset]$ vertices eliminates $p$-illusion for all vertices in the forest $T_u\setminus \{u,w_{1},... ,w_{s}\}$.

    \item \textit{Case 2:} Computing the entry $C[u,z_u,flag_u,w_q]$.
    
    We show that $C[u,z_u,flag_u,w_q]$ is a minimum $p$-illusion-free recoloring of $T_u\setminus \{u,w_{q+1},... ,w_{s}\}$. Note that by the induction hypothesis, we have already computed $C[w_q, z, 0, last(w_q)]$ and $C[w_q, z, 1, last(w_q)]$ optimally. Recall, $z$ denotes the out-neighbors of $w_q$ recolored. The minimum value of $z$ with $z\geq \text{\textit{def}}_p(w_q)$ is considered in the $1^{st}$ term (when $w_q$ is not recolored) or in the $3^{rd}$ term (when $w_q$ is recolored) of RHS of \Cref{eq:outchi}. Therefore, in both cases, we are doing a minimum recoloring. 
    Finally, by the induction hypothesis, the entries $C[u,z_u,flag_u,w_k]$ for all $w_k$ that precedes $w_q$ in $\sigma$ have been correctly computed. Thus, recoloring $C[u,z_u,flag_u,w_q]$ vertices eliminates $p$-illusion for all vertices in the forest $T_u\setminus \{u,w_{q+1},... ,w_{s}\}$.
    \end{itemize}
    \end{sloppypar}

    Finally, we compute $C[r,z_r,flag_r,last(r)]$ with $z_r\geq \text{\textit{def}}_p(r)$ ensuring $r$ is $p$-illusion-free. This completes the proof of the lemma.

    \end{proof}

    Hence, the entry $C[r,z_r,flag_r,last(r)]$ at the root vertex stores the minimum number of recoloring required for the tree $T$ to be $p$-illusion-free.

    The runtime depends on the number of entries to compute and the time needed to compute each entry. The total number of entries is at most $2\times n\times n\times n=2\times n^3$. To compute each entry, we are doing $\mathcal{O}(n)$ work. Hence, the runtime of the algorithm is $2\times n^3\times \mathcal{O}(n)=\mathcal{O}(n^4)$, polynomial in the instance size. 
\end{proof}

   \section{Parameterized Algorithms}

This section establishes the fixed-parameter tractability of $p$-\textsc{Difr} under two different parameterizations. We first present an FPT algorithm parameterized by the treewidth of the graph and the maximum deficiency. We present a PTAS for the $p$-\textsc{Difr} problem on planar graphs in Appendix ~\ref{sec:PTAS}.
%, and a result for we get a $\lambda$-outerplanar graphs. 
Subsequently, we provide FPT algorithms parameterized by the number of vertices under $p$-illusion and the treedepth of the constraints matrix of the ILP formulation of $p$-\textsc{Difr} separately. 

\subsection{FPT Algorithm Parameterized by Treewidth and Maximum Deficiency}

The NP-hardness  even on bipartite DAGs, eliminate the possibility of designing FPT algorithm parameterized by most directed width parameters including directed treewidth, directed cutwidth,  digraph bandwidth, degreewidth as all of these measure the distance from acyclicity.

We begin with the detailed construction of the treewidth-based algorithm.

%The polynomial-time algorithm for directed trees (\Cref{thm:dt}) motivates us to focus on the  parameters treewidth and maximum degree.  We show an FPT algorithm with respect to this combined parameters.
 
We design a dynamic-programming algorithm, using a nice tree decomposition of the underlying undirected graph, that tracks recoloring decisions for  each bag of the decomposition  and illusion status for the vertices in the bag to obtain the following result. For definitions of treewidth and tree decomposition see \cite{cygan2015parameterized}.

\begin{restatable}{theorem}{thmtreewidthfpt}\label{thm:treewidthfpt}
    Let $G$ be an $n$-vertex directed graph with deficiency at most $D$ and tree width \textit{tw}. Then the $p$-\textsc{Difr} problem on $G$ is solvable in time $\mathcal{O}((2 D)^{\textit{tw}})\cdot n^{\mathcal{O}(1)}$.
\end{restatable}
    
\begin{proof}

   Let $\cT=(T,\{X_t\}_{t\in V(T)})$ be a tree decomposition of the underlying undirected graph of the input $n$-vertex directed graph $G$ that has width at most $tw$ and has $p$-deficiency at most $D$.  It's a well-known result that we can get a nice tree decomposition of $G$ in polynomial time from a given tree decomposition by Lemma $7.4$ of \cite{cygan2015parameterized}. By abuse of notation, let us say that $\cT$ is a nice tree decomposition of $G$. Since we define a nice tree decomposition for rooted trees, let us assume that $T$ is rooted at some node $r$. In the nice tree decomposition of the graph, we will consider the orientations of the edges among vertices in the bags as they are in the original directed graph.

    For a node $t$ of the decomposition $T$, let $X_t$ be the set of vertices of $G$ in the bag of node $t$. Let $V_t$ be the union of all the bags present in the subtree of $T$ rooted at $t$, including $X_t$. The subgraph induced by $V_t$ can communicate with the rest of the graph only via bag $X_t$.
    
    We would like to define subproblems depending on the interaction between the solution and the bag $X_t$. 
    
    For every node $t$, every $S\subseteq X_t\cap R(V)$, and each vector $z$ of length $|X_t|$, called as the \textit{state vector}, with each component $z_v \in \{0,1, \dots,|X_t\cap R(V)|\}$, we define the following value:

$C_t[S,z]$= minimum number of vertices that need to be recolored such that \begin{enumerate}[label=(\roman*)]
    \item $S\subseteq X_t\cap R(V)$ are the only vertices recolored from $X_t$,
    \item For a vertex $v_i\in X_t$, we have $z_{v_i}$  red out-neighbors that have been recolored in $V_t$, and 
    \item every vertex in $V_t\setminus X_t$ is $p$-illusion-free.
\end{enumerate}
\medskip

In the definition of $C_t[S,z]$, the set of vertices $V_t\setminus X_t$ is $p$-illusion free means that we are ensuring the fact that the vertices in the subtree rooted at $t$ that are forgotten (i.e, not in the bag $X_t$) are $p$-illusion-free.
The vector $z$ accounts for the number of red out-neighbors recolored in $V_t$ for each vertex in the bag $X_t$. When the bag is empty, an all zero vector $z$ is denoted by $\emptyset$. We compute the entries recursively and show it correctly computes a minimum recoloring by induction on the tree decomposition.

    \textbf{Leaf node.} If $t$ is a leaf node, then $X_t=\phi$ and we have only one value $C_t[\phi,\phi]=0$.

\textbf{Introduce node.} Suppose $t$ is an introduce node with child $t'$ such that $X_t=X_{t'}\cup \{v\}$ for some $v\notin X_{t'}$. Let $S$ and $z$ be as defined above. We claim that the following formula holds:

\begin{equation}\label{eq:intronode}
    C_t[S,z]=
    \begin{cases}
        C_{t'}[S,z'] & \text{if $v\notin S$}\\
        C_{t'}[S\setminus \{v\},z']+1 &\text{if $v\in S$}
    \end{cases}
\end{equation}
where $z_v = |S \cap N^+(v)|$ and for all $u\in X_{t'}$, we have

$z_u=
\begin{cases}
    z_{u}'+1 & \text{if $v\in S\cap N^+(u)$}\\
    z_{u}' & \text{otherwise.}
\end{cases}$

To prove formally that this formula holds, consider first the case when $v\notin S$. In this case we show that the optimal value is computed by the recursive formula. As $v \notin S$, we have not recolored any new vertex when considering $X_t$. This immediately implies that $z_u'=z_u$ for each $u \in X_t$. By induction hypothesis, we know $C_{t'}[S,z']$ is given the minimum number of recolorings so far in $V_{t'}$ with $S$ being recolored in $X_{t'}$ such that $V_{t'}\setminus X_{t'}$ is $p$-illusion free. Thus, from the definition the entry $C_t[S,z]=C_{t'}[S,z']$. 

Suppose $v\in S$ and consider $C_t[S,z]$ .
First we show that the LHS is at most the value computed by the recursive computation. We have considered $S$ in $X_{t}$, so it follows that, $S'=S\setminus\{v\}$ is a set considered in the bag $X_{t'}$. As $v\notin X_{t'}$, so by including $v$ in $S'$, we increase the number of recoloring by 1 for every vertex that has $v$ as an out-neighbor in $X_{t'}$ and $z'_u$ is replaced by $z_u'+1$ for every $u$ with $v\in S\cap N^+(u)$ and otherwise remains unchanged. So $C_t[S,z]$ is at most $C_{t'}[S\setminus \{v\},z']+1$.

This gives, 

$C_t[S,z]\leq C_{t'}[S\setminus \{v\},z']+1$, 

with 
$z_u'=
\begin{cases}
    z_{u}-1 & \text{if $v\in S\cap N^+(u)$}\\
    z_{u} & \text{otherwise.}
\end{cases}$

Next we show that the recursive formula on the right hand side gives the minimum number of vertices recolored such that $V_t \setminus X_t$ is $p$-illusion free and each vertex $u$ in $X_t$ has $z_u$ recolored out-neighbors. Let $C_{t'}[S\setminus \{v\},z']$ gives a minimum recoloring with the required property for $X_{t'}$. Now as we are recoloring $v$ (that is, $v\in S$) in $X_t$, so  $C_t[S,z]-1$ is at least $C_{t'}[S\setminus \{v\},z']$ and $z$ should be consistent with $z'$. That is, $C_t[S,z]-1\geq C_{t'}[S\setminus \{v\},z']$ with

$z_u=
\begin{cases}
    z_{u}'+1 & \text{if $v\in S\cap N^+(u)$}\\
    z_{u}' & \text{otherwise.}
\end{cases}$

Thus, we show the equality in \Cref{eq:intronode}.

\textbf{Forget node.} Suppose $t$ is a forget node with child $t'$ such that $X_t=X_{t'}\setminus \{v\}$ for some $v\in X_{t'}$. We claim that the following formula holds:

\begin{equation}\label{eq:fornode}
    C_t[S,z]= \displaystyle\min_{\substack{z':z_u'=z_u\ \forall\ u\in X_{t'} \\ z_v'\geq \textit{def}_p(v) }} 
    \begin{cases}
        C_{t'}[S,z'], \\ C_{t'}[S\cup \{v\},z']
    \end{cases}
\end{equation}

Note that, when we forget $v$, we ensure that all out-neighbors and in-neighbors of $v$ have been processed and all their contributions have been locked in $z$. So, we never let deficiency of any vertex increase or decrease at forget.

We now give a formal proof of this formula. Since we forget $v$, we need to ensure that $v$ is $p$-illusion-free, as $v$ will never appear above the bag $X_{t'}$. First, consider $C_t[S,z]$, where $S\subseteq X_t\cap R(V)$ and $z$ is the corresponding state vector for $X_t$. Since we are at a forget node $v\in X_{t'}$, the DP state vector $z'$ still tracked the component $z'_v$. As $v\notin X_t$, $v$ does not appear in the state vector $z$. To compute the entry for $X_t$, we need to consider both possibilities for $v$, whether it is recolored or not in $V_{t'}$ and then take the minimum over the two possibilities. 

Note that, by definition, the minimum number of recolored vertices is stored in $C_t[S,z]$. When $S$ is   considered in the definition of $C_{t'}[S,z']$ with $z'$ that is consistent with $z$, we have $C_{t'}[S,z']\leq C_t[S,z]$ (by definition of $C_{t'}[S,z']$). If not, then we must have considered $S\cup \{v\}$ at bag $X_{t'}$, because  $X_{t'}\setminus X_t=\{v\}$ and $C_{t'}[S\cup\{v\},z']\leq C_t[S,z]$. Moreover, since we only consider the values of $z'_v\geq \textit{def}_p(v)$ in the minimization, the vertex   $v$ is $p$-illusion free in both the cases. Additionally, using induction hypothesis, each vertex in $V_t\setminus X_{t'}$ is $p$-illusion free. Therefore, when we consider $C_t[S,z]$, each vertex in $V_t \setminus X_t$ is $p$-illusion free. Thus, $C_t[S,z]$ stores the minimum number of recolored  vertices such that each vertex $u$ in $X_t$ has $z_u$ recolored out-neighbors and every vertex in $V_t\setminus X_t$ is $p$-illusion free. This shows that

$C_t[S,z]\geq \displaystyle\min_{\substack{z':z_u'=z_u\ \forall\ u\in X_{t'} \\ z_v'\geq \textit{def}_p(v) }}
    \begin{cases}
        C_{t'}[S,z'], \\ C_{t'}[S\cup \{v\},z']
    \end{cases}$

\medskip
Now we will show the other direction, that is, $C_t[S,z]$ is at most the minimum of two quantities on the right-hand side of \cref{eq:fornode} such that $V_t\setminus X_t$ becomes $p$-illusion free and each vertex $u$ in $X_t$ has $z_u$ recolored out-neighbors. Note that $C_{t'}[S,z']$ ensures minimum number of recolorings in the subtree $V_{t'}$ such that $V_{t'}\setminus X_{t'}$ becomes $p$-illusion free. As $V_t\setminus X_t=V_{t'}\setminus X_{t'}\cup \{v\}$, so to make sure that $V_t\setminus X_t$ becomes $p$-illusion free it is enough to make $v$ and  $V_{t'}\setminus X_{t'}$ $p$-illusion free. This gives that, either $C_{t}[S,z]\leq C_{t'}[S,z']$ or $C_t[S,z] \leq C_{t'}[S\cup\{v\},z']$ provided that $z'_v\geq \textit{def}_p(v)$ and $z_u'=z_u$ for all $u\in X_{t'}$. So we get,

$C_t[S,z]\leq \displaystyle\min_{\substack{z':z_u'=z_u\ \forall\ u\in X_{t'} \\ z_v'\geq \textit{def}_p(v)} }
    \begin{cases}
        C_{t'}[S,z'], \\ C_{t'}[S\cup \{v\},z']
    \end{cases}$

    Thus, we prove \Cref{eq:fornode}.

\textbf{Join node.} Suppose that $t$ is a join node with children $t_1, t_2$ such that $X_t=X_{t_1}=X_{t_2}$. The claimed recursive formula is as follows:

$C_t[S,z]= \displaystyle\min_{z^1} \{C_{t_1}[S,z^1]+C_{t_2}[S,z^2]-|S|\}$,

where \begin{equation} \label{eq:svector}
    z_u^2=z_u-z_u^1+|S\cap N^+(u)|
\end{equation} for all $u\in X_t$.

We now give a proof of this formula.

First, we will show that RHS gives a feasible  solution for $C_t[S,z]$. Take any choice of $z^1$ such that $z^2$ defined by the \Cref{eq:svector} satisfies the required constraints in $V_{t_2}$. We have taken both subtrees $V_{t_1}$ and $V_{t_2}$ to have exactly the same recoloring $S$ in $X_t$. For each $u\in X_t$, the total number of red out-neighbors recolored across the entire subtree $V_t$ is:

$z_u^1+z_u^2-|S\cap N^+(u)|
=z_u^1+(z_u-z_u^1+|S\cap N^+(u)|)-|S\cap N^+(u)|
=z_u$.

Moreover, any vertex in $V_t\setminus X_t$ belongs to either $V_{t_1}\setminus X_t$ or $V_{t_2}\setminus X_t$, and are guaranteed to be $p$-illusion free by the subproblems $C_{t_1}$ and $C_{t_2}$. So, $V_{t}\setminus X_t$ is $p$-illusion free. Therefore, the total recoloring cost is $C_{t_1}[S,z^1]+C_{t_2}[S,z^2]-|S|$, here we subtract $|S|$ to avoid double-counting of recolored vertices in $X_t$. This gives,

$C_t[S,z]\leq \displaystyle\min_{z^1} \{C_{t_1}[S,z^1]+C_{t_2}[S,z^2]-|S|\}$, with $z_u^2=z_u-z_u^1+|S\cap N^+(u)|$. 

Conversely, suppose we are given an optimal solution for $C_t[S,z]$. For each $u\in X_t$, let $z_u^1$ be the contribution from $V_{t_1}$ and $z_u^2$ from $V_{t_2}$, ensuring $z_u^1+z_u^2-|S\cap N^+(u)|=z_u$. 
This formula is valid, because, if a vertex $u\in X_t$ has an out-neighbor $a$ outside $X_t$, then the edge ($u,a$) belongs to exactly one subtree $V_{t_1}$ or $V_{t_2}$, as outside of $X_t$ the two subtrees are disjoint (the set of bags containing $v$ forms a connected subtree). So any contribution from such an edge is counted exactly once, by the subtree that contains it. If the out-neighbor $a\in X_t$, the edge ($u, a$) belongs to both the subtrees, because both $u$ and $a$ are in $X_t$. Hence it is counted twice, once in $V_{t_1}$ and once in  $V_{t_2}$. To correct this over-counting we subtract $|S\cap N^+(u)|$, where $S\cap N^+(u)$ counts the recolored out-neighbors of $u$ within $X_t$. 

Again, every vertex in $V_t\setminus X_t$ is $p$-illusion free. These vertices belongs to exactly one subtree (as bags with a common vertex are connected). This ensures that $V_{t_1}\setminus X_{t_1}$ and $V_{t_2}\setminus X_{t_2}$ are $p$-illusion free.  Hence we get,
$C_t[S,z]\geq \displaystyle\min_{z^1} \{C_{t_1}[S,z^1]+C_{t_2}[S,z^2]-|S|\}$. This completes the proof of the join node recurrence. 

By applying the formulas in a bottom-up manner on $T$ we will finally compute $C_r[\phi, \phi]$, where $r$ is the root. 

Let us analyze the runtime of this algorithm. Total runtime = (number of entries to compute) $\times$ (time taken to compute one entry). Total number of entries=$2^{\textit{tw}+1} \times {D}^{\textit{tw}+1} \times n^{\mathcal{O}(1)}$,where  $2^{\textit{tw}+1}$ denotes the number of subsets of $X_t$, ${D}^{\textit{tw}+1}$ denotes the number of vectors $z$ for $X_t$ and $n^{\mathcal{O}(1)}$ denotes the number of bags. Runtime to compute one entry is ${D}^{\textit{tw}+1}$. Therefore, the total runtime is $2^{\textit{tw}+1} \times {D}^{\textit{tw}+1} \times n^{\mathcal{O}(1)}\times {D}^{\textit{tw}+1}=\mathcal{O}(2D)^{\textit{tw}}\cdot n^{\mathcal{O}(1)}$. This completes the proof.
\end{proof}

The above theorem has applications to classes of planar graphs, called the outerplanar graphs. Thus, we shift out focus to planar graphs. We begin with the definition of $\lambda$-outerplanar graphs from \cite{bodlaender1998partial}.

\begin{definition}[$\lambda$-Outerplanar Graphs] An embedding of a graph $G = (V,E)$ is $1$-outerplanar, if it is planar, and all vertices lie on the exterior face. For $\lambda \geq 2$, an embedding of a graph $G = ( V,E)$ is $\lambda$-outerplanar, if it is planar, and when all vertices on the outer face are deleted, then a $(\lambda - 1)$-outerplanar embedding of the resulting graph is obtained. A graph is $\lambda$-outerplanar, if it has a $\lambda$-outerplanar embedding.   
\end{definition}

Moreover, the treewidth of a $\lambda$-outerplanar graph $G = (V, E)$ is at most $3\lambda- 1$ \cite{bodlaender1998partial}. We have the following corollary.

\begin{restatable}{corollary}{corkouterplanar}\label{cor:k-outerplanar}
    The $p$-\textsc{Difr} problem can be solved in polynomial-time on an $n$-vertex $\lambda$-outerplanar graphs. 
\end{restatable}
\begin{proof}
    Follows immediately from \Cref{thm:treewidthfpt}. Since the treewidth of a $\lambda$-outerplanar graph is at most $3\lambda-1$, so we can solve the $p$-\textsc{Difr} problem in time $\mathcal{O}(n^{3\lambda-1})$ (note that here the maximum degree, $\triangle\leq n$).
\end{proof}

\subsection{Formulation of ILP and the FPT algorithms}\label{sec:ILP}
We show that the $p$-\textsc{Difr} problem is fixed-parameter tractable with respect to the number of vertices under $p$-illusion. We first formulate the problem as an Integer Linear Program (ILP), which serves as the basis for developing an FPT algorithm.

\textit{The ILP formulation of the $p$-\textsc{Difr} problem}:

Let us define a binary variable for every red vertex in the graph.

Let $x_u=\begin{cases}
    1, \text{if}\ u\ \text{is recolored}\\
    0, \text{otherwise}
\end{cases}$
% We define the following integer linear programme.

\[\texttt{ILP}\text{-}p\text{-}\texttt{\textsc{Difr}: }\quad minimize\quad\sum_{u\in R(V)}  x_u\qquad\qquad\]
% \[ {\mathrm ({\tt IP-1a})}\qquad \qquad \max \quad \sum\limits_{e\in E^*(G^*)} -y_e  \label{bad:obj}
Subject to the constraints,

\begin{align}
\sum_{u\in N^+(v)} x_u \geq  \text{\textit{def}}_p(v) &\quad \forall \ v \in X_p,\label{ineq:classcontribution}\tag{$p$-illusion free}\\
x_u\in \{0, 1\}\ &\quad \forall\ u \in R(V).\label{ineq:classbound}\tag{feasibility}
\end{align}

Inequalities \eqref{ineq:classcontribution} say that among the red out-neighbors of $v$, at least $\text{\textit{def}}_p(v)$ must be recolored. The inequalities in \eqref{ineq:classbound} represent the binary choice of whether or not to recolor a given red vertex. 

\begin{restatable}{observation}{obsILPfeasibility}\label{obs:ILPfeasibility}
    The ILP formulated above correctly characterizes the $p$-\textsc{Difr} problem, and its optimal solution corresponds to an optimal recoloring.
\end{restatable}
\begin{proof}
    The \Cref{obs:pdef} implies that if in-equation \eqref{ineq:classcontribution} is satisfied for a vertex $v$, then $v$ is $p$-illusion free.  Observe that the feasibility set of the ILP is always non-empty as there always exists an $p$-illusion free recoloring by labeling all vertices to blue. Thus, by minimizing the objective function, the above ILP produces an optimal solution.
\end{proof}
    
Using the defined ILP, we obtain the following two results regarding $p$-\textsc{Difr}.

\begin{restatable}{corollary}{corpsolutionsizefpt}\label{cor:psolutionsizefpt}
    The $p$-\textsc{Difr} problem can be solved in time $2^{|X_p|\log{|X_p|}} n^{\mathcal{O}(1)}$, where $X_p$ is the set of vertices under $p$-illusion. 
\end{restatable}
\begin{proof}
We observe that our ILP formulation consists of a set of variables representing the potential recoloring of red vertices and $|X_p|$ linear constraints. A key property of this formulation is that the largest coefficient in any constraint is $1$. Following the complexity result by \cite{jansen2023integer}, an ILP with $m$ constraints and maximum coefficient $\Delta$ can be solved in $2^{m \log m + 2m \log \Delta} \cdot n^{O(1)}$ time. Substituting $m = |X_p|$ and $\Delta = 1$ into this bound yields an algorithm for $p$-\textsc{Difr} with a running time of $2^{|X_p| \log |X_p|} \cdot n^{O(1)}$, confirming the problem is in FPT parameterized by $|X_p|$.
\end{proof}

Our second application of the  ILP is using the concept of primal treedepth of the constraint matrix $A$ of the ILP. We need the following definitions (\cite{eisenbrand2019algorithmic}).

\begin{definition} \label{primalgraph}({Primal and Dual graph of $A$})

Given a matrix $A \in {\Z}^{m \times n}$, its \textit{primal graph} $G_P(A) = (V_P, E_P)$ is defined as $V_P = [n]$ and 
\[
E_P = \left\{ \{i, j\} \in \binom{[n]}{2} \middle|\ \exists k \in [m] : A_{k,i}, A_{k,j} \neq 0 \right\}.
\]
In other words, its vertices are the columns of $A$ and two vertices are connected if there is a row with non-zero entries at the corresponding columns. The \textit{dual graph} of $A$ is defined as $G_D(A) := G_P(A^\top)$, that is, the primal graph of the transpose of $A$.
\end{definition}

we always assume that $G_P(A)$ and $G_D(A)$ are connected; otherwise, $A$ has (up to row and column permutations) a block-diagonal structure 
\[
A = \begin{pmatrix} 
A_1 & & \\ 
& \ddots & \\ 
& & A_d 
\end{pmatrix}
\]
and solving the ILP amounts to solving $d$ smaller ILP instances independently.

\begin{definition}\label{treedepth}({Treedepth})
    The \textit{closure} $\text{cl}(F)$ of a rooted tree $F$ is the graph obtained from $F$ by making every vertex adjacent to all of its ancestors. We consider both $F$ and $\text{cl}(F)$ as undirected graphs. The \textit{height} of a tree $F$, denoted $\text{height}(F)$, is the maximum number of vertices on any root--leaf path. The \textit{treedepth} $\text{td}(G)$ of a connected graph $G$ is the minimum height of a tree $F$ such that $G \subseteq \text{cl}(F)$. A \textit{td-decomposition} of $G$ is a tree $F$ such that $G \subseteq \text{cl}(F)$. A td-decomposition $F$ of $G$ is \textit{optimal} if $\text{height}(F) = td(G)$. 
\end{definition}
The treedepth $td$ of a graph $G$ with an optimal td-decomposition $F$ can
be computed in time $2^{td^2}|V(G)|$ \cite{reidl2014faster}.

\begin{proposition}\label{prop:treedepthFPT}{\cite{eisenbrand2019algorithmic}}
    Let $I$ be an instance of an ILP given by $A\in \Z^{m\times n}$ and $td(G_P)$ be the treedepth of $G_P(A)$. Then, $I$ is solvable in time FPT parameterized by $td(G_P)$.
\end{proposition}
Since our constraint matrix $A$ in \texttt{ILP}-$p$-\texttt{\textsc{Difr}} is a binary matrix, we use the above proposition to obtain the following.
\begin{restatable}{corollary}{cortreedepthfpt}\label{cor:treedepthfpt}
    $p$-\textsc{Difr} can be solved in time FPT parameterized by the primal treedepth of \texttt{ILP}-$p$-\texttt{\textsc{Difr}}.
\end{restatable}
\begin{proof}
    Follows immediately from \Cref{prop:treedepthFPT}.
\end{proof}

\section{Conclusion}
In this work, we introduce and study majority illusion elimination on directed graphs, by recoloring a minimum number of vertices so that no vertex remains under illusion. We establish that the problem is NP-complete on grids and contrast it with the positive result on an oriented grids, called outward grid. We extend the concept to $p$-illusion and show it remains W[2]-hard even on DAGs when deficiencies are one. On the other hand, we show a polynomial time algorithm for trees and a FPT algorithm parameterized by treewidth and deficiency. We also showed that the problem is hard on bounded degree graphs. Our work opens several directions for future research. First, is it true that $p$-\textsc{Difr} is NP-complete for $p<\frac{\Delta-1}{\Delta}$ for each $\Delta>3$? 
% Second, can $p$-\textsc{Difr} admit FPT algorithm or polynomial kernels with respect to treewidth? 
Second, what are the best possible approximation algorithms for general graphs? Finally, which additional structures can yield more efficient algorithms, specifically for $p$-illusion elimination?

\bibliographystyle{plain}
\bibliography{arxiv}

\newpage
\appendix

\section{Appendix}

We have a \textsc{Ptas} for the $p$-\textsc{Difr} problem on  planar graphs.
\subsection{A PTAS for Planar Graphs}\label{sec:PTAS}
We use the classical layering technique introduced by \cite{baker1994approximation} for designing approximation algorithms on planar graphs. On a high level, we break the input graph into layers to solve the problem optimally in each layer. Then we take a union of these solutions to return a feasible solution for the original input. The proof of the following theorem follows directly from \cite{fioravantes2025eliminating}. We have a similar result for the \textsc{Difr} problem.

\begin{restatable}{theorem}{thmptas}\label{thm:ptas}
For any given $\varepsilon > 0$, there exists an approximation algorithm which computes the
solution of the $p$-\textsc{Difr} problem on a planar graph in $O(n^{1/\varepsilon})$ time and computes a solution
which is $(1 + \varepsilon)$-times the optimum.
\end{restatable}

\begin{proof}

We have the following algorithm for planar graphs.
\begin{algorithm}
\caption{PTAS on Planar Graphs}
\label{alg:ptas}
\begin{algorithmic}[1]
\Require Directed Graph $G$, function $f$, parameter $\varepsilon$
\State $k \gets 4 / \varepsilon$;
\State Find the outerplanar layers of the vertices
\For{$i = 0$ \textbf{to} $k - 1$}
    \State Find the components $G^i_1, G^i_2, \ldots$ of $G$, where each $G^i_j$ contains vertices on the layers from $i + j(k + 1)$ to $i + (j + 1)(k + 1) + 1$
    \For{$j = 0, 1, 2, \ldots$}
        \State Compute $A^i_j$, the solution of \textsc{Difr} on $G^i_j$
        \Statex \hspace{\algorithmicindent} setting $S = \bar{G}^i_j$, where $\bar{G}^i_j$ contains vertices on layers from
        $i + j(k + 1) + 1$ to $i + (j + 1)(k + 1)$
        \State $A^i \gets \bigcup_j A^i_j$;
    \EndFor
\EndFor
\State Let $A$ be the minimum solution among $\{A^0, A^1, \ldots, A^{k-1}\}$;
\State \Return $A$
\end{algorithmic}
\end{algorithm}

We show the correctness of \Cref{alg:ptas}.

\textbf{Correctness.}
Let $A^i_j$ be the solution of the problem on $G^i_j$ for each $i,j$ using \Cref{thm:treewidthfpt}. Let $\widehat{G}^i_j \subseteq G^i_j$ be the induced subgraph on the vertices of layers
$i + j(k + 1) + 2$ to $i + (j + 1)(k + 1) - 1$.
Let $\widehat{A}^i_j = A^i_j \cap \widehat{G}^i_j$.
Let $S$ be the minimum set of vertices that must be relabeled to eliminate all the illusions in $G$.
Let $\widehat{S}^i_j = S \cap \widehat{G}^i_j$.
Clearly, $\widehat{A}^i_j =\widehat{S}^i_j$ (because \Cref{thm:treewidthfpt} gives optimal solution on $G_j^i$).
Let $V^i_j \subseteq V$ be the vertices on layers $i + j(k + 1)$ and $i + j(k + 1) + 1$.
Let $S^i_j = S \cap V^i_j$ and $S^i = \bigcup_j S^i_j$.
Observe that $\widehat{S}^i_j$ and $S^i$ are disjoint and
$\sum_{i=0}^{k-1} |S^i| = 2|S|$ (because $S^i$ and $S^{(i+1)}$ have common layer $(i+1)+j(k+1)$).
Hence, there exists one $i^*$ such that $|S^{i^*}| \leq 2|S|/k$. Let $A = \bigcup_j A^{i^*}_j$ be the union
of the solution of the $G^{i}_j$'s. Then, $A$ is a potential candidate set of vertices for relabeling for the graph $G$, as for a fixed $i$ we get that the $G^i_j$'s cover the entire graph $G$. Let $A_j = A \cap V^{i^*}_j$ be the output of the algorithm restricted to the vertices on layers $i^* + j(k + 1)$ and $i^* + j(k + 1) + 1$.

We claim that for every $j$, we have that $|A_j| \leq 2|S^{i^*}_j|$. Indeed, the set $S^{i^*}_j$ contains the subset of the optimum solution restricted to the vertices on layers $i^* + j(k + 1)$ and $i^* + j(k + 1) + 1$. Hence, $S^{i^*}_j$ can be seen as the union of two sets
$\widehat{S}^{i^*}_j$ and $\widehat{S}^{i^*}_{j'}$
where $\widehat{S}^{i^*}_j$ contains the vertices on
layer $i^* + j(k + 1)$ and $\widehat{S}^{i^*}_{j'}$ contains the vertices on layer $i^* + j(k + 1) + 1$.
Moreover, the set $A_j$ can be seen as a union of four sets:
$A^{i^*}_j$, $A^{i^*}_{j'}$, $A^*_j$ and $A^*_{j'}$, where
\begin{itemize}
    \item $A^{i^*}_j$ contains vertices from the layer $i^* + j(k + 1)$ that appear in the solution of the piece $G^{i^*}_{j-1}$; and
    \item $A^{i^*}_{j'}$ contains vertices from the layer $i^* + j(k + 1) + 1$ that appear in the solution of the piece $G^{i^*}_j$; and
    \item $A^{*}_j$ contains vertices from the layer $i^* + j(k + 1)$ that appear in the solution of the piece $G^{i^*}_j$; and
    \item $A^{*}_{j'}$ contains vertices from the layer $i^* + j(k + 1) + 1$
    that appear in the solution of the piece $G^{i^*}_{j-1}$.
\end{itemize}

Now, out-neighbors of a vertex at some layer $i$
can only be present in layers $i - 1$, $i$, and $i + 1$.
Hence, vertices of $A^{i^*}_j$ are essential for the solution of $G^{i^*}_{j-1}$ and vertices of $A^*_j$ are essential for the solution of $G^{i^*}_j$. Thus, $|A^*_j \cup A^{i^*}_j| \leq 2|\widehat{S}^{i^*}_j|$. Following a similar argument, we can show that
$|A^*_{j'} \cup A^{i^*}_{j'}| \leq 2|\widehat{S}^{i^*}_{j'}|$. Together, these two inequalities imply that $|A_j| \leq 2|S^{i^*}_j|$.

To finish the proof, it suffices to set $k = \frac{4}{\varepsilon}$.
As discussed above, a solution using our
algorithm is a feasible solution for the problem.
Consider the solution $A$ which is minimum among different choices of $i$.
It holds that
\[
|A| = \sum_j |A^{i^*}_j| + \sum_j |A_j|
\leq |S| + \frac{4|S|}{k}
= (1 + \varepsilon)|S|.
\]
Hence the theorem.
\end{proof}
 
\end{document}